\documentclass[11pt]{article}
\usepackage{times}
\usepackage{setspace}
\usepackage{amsfonts}
\usepackage{amsmath}
\usepackage{amsthm}
\usepackage{amssymb}
\usepackage{oldgerm}
\usepackage[]{graphicx}

\newcounter{aid}
\newcommand{\appsec}{
\setcounter{aid}{\value{section}}
\renewcommand{\theaid}{\Alph{aid}}
\renewcommand{\thesection}{Appendix\ \Alph{section}}
\section{\ }
\renewcommand{\thesection}{\Alph{section}}
}

\theoremstyle{plain}
\newtheorem{theorem}{Theorem}[section]
\newtheorem{corollary}[theorem]{Corollary}
\newtheorem{proposition}[theorem]{Proposition}
\newtheorem{lemma}[theorem]{Lemma}

\newtheorem{definition}[theorem]{Definition}

\theoremstyle{definition}
\newtheorem{remark}[theorem]{Remark}

\DeclareMathOperator{\supp}{supp}
\DeclareMathOperator{\card}{card}
\DeclareMathOperator{\diam}{diam}

\def\R{\mathbb{R}}

\begin{document}

\bibliographystyle{ieeetralpha}

\title{High-resolution scalar quantization with R\'enyi entropy constraint}
{\renewcommand{\thefootnote}{}
\footnotetext{\hspace{-0.3cm} W.\ Kreitmeier is with the 
Department of Informatics and Mathematics, University of Passau, 
Innstra\ss e 33, 94032 Passau, Germany (email:
\texttt{wolfgang.kreitmeier@uni-passau.de}). 
T.  Linder is with the Department of
Mathematics and Statistics, Queen's University, Kingston, Ontario,
Canada K7L 3N6 (email: {\tt linder@mast.queensu.ca}). }

\footnotetext{\hspace{-0.3cm} This research was supported in part by
  the German Research Foundation (DFG) and the Natural Sciences and
  Engineering Research Council (NSERC) of Canada.} }

\author{Wolfgang Kreitmeier and Tam\'as Linder} 
\date{July 5, 2011}

\maketitle

\begin{abstract}  
We consider optimal scalar quantization with $r$th power distortion
and constrained R\'enyi entropy of order $\alpha$. For sources with 
absolutely continuous distributions the high rate asymptotics of the
quantizer distortion has long been known for $\alpha=0$ (fixed-rate
quantization) and $\alpha=1$ (entropy-constrained quantization).
These results have recently been extended to quantization with R\'enyi
entropy constraint of order $\alpha \ge r+1$. Here we consider the
more challenging case $\alpha\in [-\infty,0)\cup (0,1)$ and for a
  large class of absolutely continuous source distributions we
  determine the sharp asymptotics of the optimal quantization
  distortion. The achievability proof is based on finding
  (asymptotically) optimal quantizers via the companding approach, and
  is thus constructive.

\end{abstract}


\allowdisplaybreaks

\section{Introduction}

With the exception of a few very special source distributions the
exact analysis of the performance of optimal quantizers is a
notoriously hard  problem. The asymptotic
theory of quantization facilitates such  analyses by assuming that
the quantizer operates at asymptotically high rates.  The seminal work
by Zador \cite{Zad63} determined the asymptotic behavior of the
minimum quantizer distortion under a constraint on either the
log-cardinality of the quantizer codebook (fixed-rate quantization) or
the Shannon entropy of the quantizer output (entropy-constrained
quantization).  (See the article by Gray and Neuhoff
\cite{GrNe98} for a historical overview and related results.) Zador's
results were later clarified and generalized by Bucklew and Wise
\cite{BuWi82} and Graf and Luschgy \cite{GrLu00} for the fixed-rate
case, and by Gray \emph{et al}.\ \cite{GrLiLi02} for the
entropy-constrained case.

Recently, approaches that incorporate both the fixed and
entropy-constrained cases have been suggested. In \cite{GrLiGi08} a
Lagrangian formulation is developed  which puts a simultaneous
constraints on entropy and codebook size, including fixed-rate and
entropy-constrained quantization as special cases. Another approach
that has been suggested in \cite{GrLiGi08} and further developed in
\cite{Kre10b,Kre10a} uses the R\'enyi entropy of order $\alpha$  of the
quantizer output as (generalized) rate. One obtains fixed-rate
quantization for $\alpha = 0$, while  $\alpha = 1$ yields the usual
(Shannon) entropy-constrained quantization approach. 

The choice of R\'enyi entropy as the quantizer's rate can be motivated
from a purely mathematical viewpoint.  In the axiomatic approach to
defining entropy, R\'enyi's entropy is a canonical extension of
Shannon-entropy, satisfying fewer of the entropy axioms
\cite{Ren60b,AcDa75}.  From a more practical point of view, the use of
R\'enyi entropy as quantizer rate is supported by Campbell's work
\cite{Cam65}, who considered variable-length  lossless codes with  exponentially weighted average
codeword length and showed that R\'enyi's entropy plays an analogous
role to Shannon entropy in this more general setting. Further results
on lossless coding for R\'enyi entropy were obtained in \cite{Nat75}.
Jelinek \cite{Jel68} showed that R\'enyi's entropy (of an appropriate
order $\alpha \in (0,1)$) of a variable-length lossless code
determines the encoding rate for a given reliability (exponential
decrease of probability) of buffer overflow when the codewords are
transmitted over a noiseless channel at a fixed per symbol rate. At
least in such situations, measuring the quantizer's rate by R\'enyi's
entropy is operationally justified.  An overview of related results
can be found in \cite{Bae03}.  The diverse uses of R\'enyi's entropy
(and differential entropy) in emerging fields such as quantum
information theory (e.g.\ \cite{Hol06}), statistical learning
(e.g.\ \cite{Jen09}), bioinformatics (e.g.\ \cite{KrMaKa04}), etc.,
may also provide future motivation for this rate concept.

The only available general result on quantization with R\'enyi entropy
constraint appears to be \cite{Kre10b} where the sharp asymptotic
behavior of the $r$th power distortion of optimal $d$-dimensional
vector quantizers has been derived for $\alpha \in [1+r/d, \infty
]$. The proof shows that for these $\alpha$ values the optimal
quantization error is asymptotically determined by the distortion of a
ball with appropriate radius around the most likely values of the
source distribution. Thus  it suffices to evaluate the $r$th
moment of this ball (see \cite[Theorem 4.3]{Kre10b}), which remarkably
simplifies the derivation and makes  the case $\alpha \ge 1+
r/d$  quite unique. In the classical ($\alpha=0$ and $\alpha=1$)
settings, the contributions of the codecells of an optimal quantizer
to the overall distortion are asymptotically of the same order. Bounds
on the optimal performance in \cite{Kre10b} suggest a similar situation
for $\alpha < 1+r/d$, making the problem more challenging
than the case $\alpha \ge 1+r/d$.

In this paper, at the price of restricting the treatment to the scalar
($d=1$) case, we are able to determine the asymptotics of the optimal
quantization error under a R\'enyi entropy constraint of order $\alpha
\in [-\infty, 0)\cup (0,1)$ for a fairly large class of source
  densities. The achievability part of the proof (providing a sharp
  upper bound on the asymptotic performance) is constructive via
  companding quantization. In particular, we determine the optimal
  point density function for each $\alpha \in [-\infty, 1+r)$ and
    provide rigorous performance guarantees for the associated
    companding quantizers (for $\alpha=0$ and $\alpha=1$, these
    results have of course been known). Matching lower bounds are
    provided for $\alpha \in [-\infty,0)\cup (0,1)$, which leaves only
      the case $\alpha \in (1,1+r)$ open.  We note that in proving the
      matching lower bounds, one cannot simply apply the techniques
      established for $\alpha=0$ or $\alpha=1$. In our case the
      distortion and R\'enyi entropy of a quantizer must be
      simultaneously controlled, a difficulty not encountered in
      fixed-rate quantization. Similarly, the Lagrangian formulation
      that facilitated the corrected proof of Zador's
      entropy-constrained quantization result in \cite{GrLiLi02}
      cannot be used since it relies on the special functional form of
      the Shannon entropy.  On the other hand, using the monotonicity
      in $\alpha$ of the optimal quantization error, one can show that
      our results imply the well-known asymptotics for $\alpha \in
      \{0,1\}$, at least for the special class of scalar distributions
      we consider.

The paper is organized as follows. In Section~\ref{sec_notation} we
introduce the quantization problem under a R\'enyi entropy constraint
and review some definitions and notation. In Section~\ref{secmain},
after summarizing some related work, we state our main result. The
next three sections are devoted to developing the machinery needed in
the proof. Section~\ref{sec_comp} presents results on the asymptotic
distortion and R\'enyi entropy of companding quantizers, which, with
the proper choice of the compressor function in a Bennett-like
integral, will turn out to be (asymptotically) optimal. In
Section~\ref{sec_fund} technical results needed mostly for
establishing lower bounds are developed. Section~\ref{sec_mix}
presents upper and lower bounds on the optimal quantization error for
mixture distributions. Section~\ref{sec_proof} contains the proof of
the main results. Section~\ref{sec_concl} contains concluding remarks
and a discussion on extending the results to vector quantization.  All
the longer, technical proofs of the auxiliary results are relegated to
the appendices.

\section{Preliminaries and notation}
\label{sec_notation}

We begin with the definition of R\'enyi entropy of order $\alpha$.

\begin{definition}
\label{def:renyi_entr}
Let $\mathbb{N} := \{1,2,\ldots \}$. Let $\alpha \in \thinspace [-
  \infty, \infty ]$ and $p=(p_{1},p_{2},\ldots) \in
[0,1]^{\mathbb{N}}$ be a probability vector, i.e.,
$\sum_{i=1}^{\infty} p_{i} = 1$.  The R\'enyi entropy of order
$\alpha$, $\hat{H}^{\alpha}(p) \in {} [0,\infty ]$, is defined as (see
\cite{Ren60b}, \cite[Definition 5.2.35]{AcDa75} and
\cite[p.\ 1]{Har09})
\begin{equation}
\label{renyi_entr}
\hat{H}^{\alpha}(p)=
\begin{cases}
\frac{1}{1-\alpha} \log \left(  \sum\limits_{i: p_i>0} p_{i}^{\alpha} \right), 
& \alpha \in {} (- \infty , \infty) \setminus \{1\}\\
- \sum\limits_{i=1}^{\infty } p_{i} \log p_{i}, & \alpha = 1 \\ 
- \log \left( \max \{ p_{i} : i \in \mathbb{N} \} \right), & \alpha = \infty \\
- \log \left( \inf \{ p_{i} : i \in \mathbb{N}, p_{i}>0 \} \right), &
 \alpha = - \infty. 
\end{cases}
\nonumber
\end{equation}
\end{definition}

We use the conventions $0 \cdot \log 0 :=0$ and $0^0:=0$. All logarithms are to the base $e$.

\begin{remark}
\label{remark_hospital}
(a) With these conventions we obtain 
\[
\hat{H}^{0}(p) = \log \left(  \card \{ i \in \mathbb{N}: p_{i} > 0 \} \right),
\]
where $\card$ denotes cardinality.
Using l'Hospital's rule it is easy to see, that the case $\alpha = 1$
follows from the case  $\alpha \neq 1$ by taking the limit $\alpha
\rightarrow 1$. 
(see, e.g.,  \cite[Remark 5.2.34]{AcDa75}).
Moreover, one has 
\begin{equation}
\label{limitdef}
\lim_{\alpha \rightarrow \infty } \hat{H}^{\alpha}( p ) =
\hat{H}^{\infty}( p ) 
\quad \text{and} \quad
\lim_{\alpha \rightarrow - \infty } \hat{H}^{\alpha}( p ) = \hat{H}^{-
  \infty}( p ). 
\end{equation}
(b) We note that the usual definition of  R\'enyi entropy is
restricted to nonnegative values of the order $\alpha$. However, it
will turn out that the case 
$\alpha<0$ can be handled without too much additional technical
difficulties, and we believe that this generalization may turn out to
have useful implications. 
\end{remark}

Now let $d \in \mathbb{N}$ and $X$ be an $\mathbb{R}^{d}$-valued random variable with 
distribution $\mu$.
Let $\mathbb{I} \subset \mathbb{N}$ and $\mathcal{S} = \{ S_{i} : i
\in \mathbb{I}  \}$ be a countable and Borel measurable partition
of $\mathbb{R}^{d}$. Moreover let $\mathcal{C} = \{ c_{i}  : i \in
\mathbb{I} \}$ be a countable set of distinct
points in $\mathbb{R}^{d}$. Then $(\mathcal{S}, \mathcal{C} )$ defines
a \emph{quantizer}  
$q : \mathbb{R}^{d} \rightarrow \mathcal{C} $ such that
\[
q(x) = c_{i} \qquad \text{ if and only if } \qquad x \in S_{i} .
\]  
We call $\mathcal{C}$ the \emph{codebook} and the $c_i$ the
codepoints. Each $S_{i} \in \mathcal{S}$ is called
\emph{codecell}. Clearly, $\mathcal{C} = q (\mathbb{R}^{d})$ (the
range of $q$). Moreover, 
\[
\mathcal{S} = \{ q^{-1}(z) : z \in q(\mathbb{R}^{d})  \}
\]
where $q^{-1}(z)=\{x\in \R^d: q(x)=z\}$.
Let $\mathcal{Q}_{d}$ denote the set of all quantizers on $\R^d$,
i.e., the set of all Borel-measurable mappings $q : \mathbb{R}^{d}
\rightarrow \mathbb{R}^{d}$ with a countable number of codepoints $q(\mathbb{R}^{d})$.  The discrete random variable $q(X)$ is a quantized
version of the random variable $X$ whose distribution is denoted by $\mu \circ q^{-1}$. In measure-theoretical terms the
image measure $\mu \circ q^{-1}$ has a countable support and defines
an approximation of $\mu$, the so-called quantization of $\mu$ by $q$.
With any enumeration $\{ i_{1}, i_{2}, \ldots \}$ of $\mathbb{I}$ we define
\begin{equation}
\label{ref_def_entr}
H^{\alpha}_{\mu }(q) = 
\hat{H}^{\alpha } ( \mu ( S_{i_{1}} ), \mu ( S_{i_{2}} ),\ldots ) 
\end{equation}
as the R\'enyi entropy of order $\alpha$ of $q$ with respect to $\mu$.
We intend to quantify the error in approximating the original
distribution $\mu$ with its quantized version $\mu \circ q^{-1}$.  To
this end let $\parallel \cdot \parallel$ be any norm on
$\mathbb{R}^{d}$ and $\rho: [0,\infty) \rightarrow [0,\infty)$ a
    strictly increasing function. For $q \in \mathcal{Q}_{d}$ we
    measure the approximation error between $X$ and $q(X)$,
    resp.\ $\mu$ and $\mu \circ q^{-1}$, also called the quantizer
    distortion,  as
\begin{equation}
\nonumber
\label{ref_dmuf}
D_{\mu}(q)= E \rho (\| X-q(X)\| ) =
\int \rho ( \| x-q(x)\| )\,  d \mu (x).
\end{equation} 
For any $R \geq 0$ we define
\begin{equation}
\label{ref_darst_f_q}
D_{\mu}^{\alpha }(R)= \inf \{ D_{\mu}(q) : q \in \mathcal{Q}_{d}, H^{\alpha }_{\mu }(q) \leq R \},
\end{equation}
the optimal quantization distortion of $\mu$ under R\'enyi $\alpha$-entropy
bound $R$. We note that $D_{\mu}^{\alpha }(R)$ is a nonincreasing
function of $\alpha$ (see Lemma~\ref{lemm_mono}).

We call a quantizer $q$ optimal for $\mu$ under the entropy constraint
$R$ if $D_{\mu}(q)=D_{\mu}^{\alpha }(R)$ and $H^{\alpha }_{\mu }(q)
\le R$. In the rest of this paper we focus on the one-dimensional case
(scalar quantizers, $d=1$) and the so-called $r$th power distortion
measure $\rho (x) = x^{r}$, where $r \ge 1$. Thus the distortion of
quantizer $q\in \mathcal{Q}_1$ is given by
\[
D_\mu(q)=
E|X-q(X)|^r = \int | x-q(x)|^r\,   d \mu (x).
\]

For simplicity we write $\mathcal{Q}_{1}=\mathcal{Q}$. Also, let
$\mathcal{Q}^c\subset \mathcal{Q}$ denote the set of all scalar
quantizers with finitely many codecells, each of which is an interval,
and such that every codepoint lies in the closure of the corresponding
codecell. The following lemma (proved in \ref{appA}) presents two key
properties of optimal quantization under R\'enyi entropy constraint.

\begin{lemma}
\label{lemm_mono}
For all $R \geq 0$ and $\alpha,\beta  \in [- \infty, \infty ]$ with  $\beta\leq \alpha$,  we have 
\begin{equation}
\label{eq:mono}
 D_{\mu}^{\beta }(R) \ge D_{\mu}^{\alpha }(R).
\end{equation} 
Assume  that  $E|X|^r<\infty$ and $\mu$ is nonatomic. Then  for all $R \geq 0$
and $\alpha \in [-\infty,0]$, we have 
\begin{equation}
\label{eq:interval}
D_{\mu}^{\alpha }(R)= \inf \{ D_{\mu}(q) : q \in \mathcal{Q}^c,
H^{\alpha }_{\mu }(q) \leq R \}
\end{equation} 
while for all $\alpha\in (0,\infty]$,
\begin{equation}
\label{eq:interval1}
D_{\mu}^{\alpha }(R)= \inf \{ D_{\mu}(q) : q \in \mathcal{Q}^c,
H^{\alpha }_{\mu }(q) = R \}.
\end{equation} 
\end{lemma}

The second statement of the lemma says that under the given conditions
the optimum quantizer performance can be approached arbitrarily
closely by quantizers in $\mathcal{Q}^c$. For this reason, in the rest
of the paper all quantizers will be assumed to belong to
$\mathcal{Q}^c$; in particular, we only consider quantizers with
finitely many interval cells. According to
(\ref{eq:interval1}),  when $\alpha\in (0,\infty]$  it suffices to consider only those
quantizers in $\mathcal{Q}^c$ whose entropy attains $R$.

From \cite[Thm.\ 5.2]{Kre10b} it is known that for $\alpha \in [0,1]$ the
product $e^{rR}D_{\mu}^{\alpha}(R)$ remains bounded and is bounded away
from zero as $R\to \infty$. This motivates the following notion of
quantizer optimality that will play an important role in our work.

\begin{definition}
Let $(q_{n})_{n \in \mathbb{N}} \subset \mathcal{Q}$ be
a sequence of quantizers such that 
$H_{\mu}^{\alpha }(q_{n})\to \infty$  as $n \to \infty$. 
If $e^{rR}D_{\mu}^{\alpha}(R)\to c$ as  $R\to \infty$ for some $c \in
(0,\infty)$ and 
\begin{equation}
\label{def_asymp_opt_quant}
\lim_{n \to \infty} e^{rH_{\mu}^{\alpha }(q_{n})} D_{\mu }(q_{n}) = c, 
\end{equation}
then we call $(q_{n})_{n \in \mathbb{N}}$ an asymptotically optimal 
sequence of quantizers for $\mu$.
\end{definition}

We denote by $\lambda$ the one-dimensional Lebesgue measure.
For a measurable real function $f$ on $\mathbb{R}$ 
and measurable nonempty set $A \subset \mathbb{R}$, $\mathrm{ess } \inf
\nolimits_{A} f = \sup \{b: \lambda ( \{x \in A : f(x) < b \} ) = 0 \} 
$ denotes that the essential infimum of $f$ on $A$. Similarly,
$\mathrm{ess } \sup \nolimits_{A} f =\inf \{b: \lambda ( \{x \in A :
f(x) > b \} ) = 0 \}$  is the essential supremum of $f$ on $A$. We let $\supp (\mu )$ denote  the support of $\mu$ 
defined by
\[
\supp(\mu)=\{ x:\, \mu((x-\epsilon, x+\epsilon))>0 \text{ for all
          $\epsilon>0$}\}. 
\]
Note that $\supp(\mu)$ is the smallest closed set whose complement has
$\mu$ measure zero. We will often deal with  the situation where
$\supp(\mu)$ is contained in a bounded interval $I$. In such cases, we
usually  leave a  quantizer $q\in \mathcal{Q}$ undefined outside $I$, as we may
since $\mu(\R\setminus I)=0$. 

Let $\mathbb{Z}$ denote the set of all integers and assume $\Delta >
0$.  The infinite-level uniform quantizer $\hat{q}_{\Delta}$ on $\mathbb{R}$
has  codecells $\{ (i\Delta,(i+1)\Delta] : i \in
  \mathbb{Z} \}$ and corresponding  codepoints that are the midpoints
  of the associated cells, so that $\hat{q}_{\Delta}(x)=(i+1/2)\Delta$ if and only if $x\in
  (i\Delta,(i+1)\Delta]$.

\section{Main results}
\label{secmain}

First we summarize the known results regarding the sharp high-rate
asymptotics of the distortion of optimal scalar quantizers. In order to
unify the treatment, we reformulate the classical (resolution and
entropy) rate constraints in terms of the R\'enyi entropy with
appropriate order. For $r > 0$ we let
\[
C(r) = \frac{1}{(1+r)2^{r}} .
\]

\begin{theorem}[\cite{Zad63, BuWi82, GrLu00, GrLiLi02,
      Kre10b}]  
\label{ref_theo_main_resul_known}
Let $r \ge  1$ and $\mu = \mu_{a} + \mu_{s}$ be the Lebesgue
decomposition of  distribution $\mu$ of the scalar random variable $X$
with respect
to the one-dimensional Lebesgue measure $\lambda$, where $\mu_{a}$ denotes the absolutely continuous part
and $\mu_{s}$ the singular part of $\mu$.  Assume that $\mu_{a}(\R)>0$
and let $f=\frac{d\mu_{a}}{d\lambda}$  be the density of $\mu_a$. 
\begin{itemize}
\item[(i)] If $\alpha = 0$ and $E|X|^{r+\delta} < \infty$
for some $\delta > 0$,
then
\begin{equation}
\label{main_asymp}
\lim_{R \rightarrow \infty } e^{rR} D_{\mu }^{0}(R)  = 
C(r) \left( \int  f^{1/(1+r)}\, d \lambda \right)^{1+r}.
\end{equation}
\item[(ii)] If $\alpha = 1$, $\mu_{s}(\R)=0$,
$\int f \log  f\,  d \lambda$ exists and is finite, and
$H_{\mu}^{1}(\hat{q}_\Delta) < \infty$ for some $\Delta>0$, then
\begin{equation}
\label{main_asympalpha1}
\lim_{R \rightarrow \infty } e^{rR} D_{\mu}^{1}(R) = C(r) e^{-r \int f
  \log  f  \, d \lambda}.
\end{equation}
\item[(iii)] If $\alpha \in [1+r, \infty]$, $\mu_{s}(\R)=0$,
  $E|X|^{r+\delta} < \infty$ for some $\delta > 0$, and $\mathrm{ess }
  \sup _{\mathbb{R}} f < \infty$, then
\[
\lim_{R \rightarrow \infty }e^{(1+r)\beta(\alpha )R} D_{\mu}^{\alpha
}(R) = C(r)  \left( \mathrm{ess } \sup{} _{\mathbb{R}} f  \right)^{-r},
\]
where $\beta(\alpha)=(\alpha-1)/\alpha$ if $ \alpha\in [1+r,\infty)$
  and $\beta(\alpha)=1$ if $ \alpha=\infty$.
\end{itemize}
\end{theorem}

Note that $f$ is a probability density function if and only if
$\mu_s(\R)=0$. Part (i) of the theorem is originally due to Zador
\cite{Zad63} who considered the multidimensional case; corrected and
generalized proofs were given by Bucklew and Wise \cite{BuWi82} and
Graf and Luschgy \cite{GrLu00}. Part (ii) is also due to Zador
\cite{Zad63} with corrections and generalizations by Gray \emph{et
  al}.\ \cite{GrLiLi02}. Part (iii) is due to Kreitmeier \cite{Kre10b}
who also gave upper and lower bounds for the case $\alpha\in (1,1+r)$.

\begin{definition}
\label{def_weak_unimod}
A one-dimensional probability density function $f$ is called weakly
unimodal if $f$ is continuous on its support and there exists an
$l_{0}>0$ such that $\{x: f(x)\ge l\}$ is a compact interval for every
$l \in (0, l_{0})$.
\end{definition}

\begin{remark} Note if $f$ is weakly unimodal density, then it  is
  bounded and its support is a (possibly unbounded) interval. Clearly,
  all continuous unimodal densities are weakly unimodal. Thus the
  class of weakly unimodal densities includes most parametric source
  density classes commonly used in modeling information sources such
  as exponential, Laplacian, Gaussian, and generalized Gaussian
  densities.
\end{remark}

For $\alpha \in {} (-\infty,r+1) \setminus \{ 1 \}$ we define
\begin{equation}
\label{lambda12def}
a_1 = \frac{1- \alpha + \alpha r}{1- \alpha + r} , \qquad
a_2 =  \frac{1- \alpha + r}{1- \alpha } .
\end{equation}
The following is the main result of the paper. 

\begin{theorem} 
\label{ref_theo_main_resul}
Let $r > 1$ and assume that the distribution $\mu$ of $X$ is
absolutely continuous with respect to $\lambda$ having density $f$.  Assume
that $\mathrm{ess } \sup \nolimits_{\mathbb{R}} f < \infty$ and let $M=
 (\inf ( \supp (\mu ) ) ), \sup ( \supp (\mu ) ))$.  In either of
the following cases:
\begin{itemize}
\item[(i)] $\alpha \in ( 0,1 )$, $E|X|^{r+\delta} <
  \infty$ for some $\delta > 0$, and $f$ is weakly unimodal,
\item[(ii)]  $\alpha \in ( -\infty, 0 )$, $ \mathrm{ess } \inf
  \nolimits_{M} f >0$ and $f$   is continuous on $M$,
 \end{itemize}
we have 
\begin{equation}
 \label{ref_equ_erdmuarf1} 
\lim_{R \rightarrow \infty } e^{rR} D_{\mu }^{\alpha }(R) 
 =   C(r) \left( \int_{M} f^{a_1 }\, d \lambda
\right)^{a_2}.
\end{equation}
If  $ \mathrm{ess } \inf \nolimits_{M} f>0$ 
and $f$ is continuous on $M$, then
\begin{equation}
\label{sharp_asymp_alph_neg_inf}
\lim_{R \rightarrow \infty } e^{rR} D_{\mu }^{- \infty }(R) 
=  C(r) \left( \int_{M} f^{1-r}\,  d \lambda \right).  
\end{equation}
\end{theorem}

The proof of the theorem is given in Section~\ref{sec_proof}. Upper
bounds will be established using a companding approach, while matching
lower bounds are developed by considering increasingly more general
classes of source densities. 

\begin{remark} \label{remH} (a) Note that if we formally substitute
  $\alpha=0$ in (\ref{ref_equ_erdmuarf1}), it reduces to
 (\ref{main_asymp}).  Moreover, it is easy to show that 
 (\ref{ref_equ_erdmuarf1}) reduces to (\ref{main_asympalpha1}) if
 $\alpha \to 1$. Due to monotonicity of the 
  quantization error (Lemma~\ref{lemm_mono}) and by the
  upper bound   for the quantization error for $\alpha \in [-\infty,
    1+r)$ (Corollary~\ref{asymp_quant_comp})
 one can rigorously show that the known asymptotics for $\alpha \in
 \{0,1\}$ also follow from Theorem~\ref{ref_theo_main_resul}, at
 least in the scalar case and 
under our restrictions on  the source density.

\noindent (b) The results of the theorem can be expressed in terms of  the R\'enyi differential entropy $h^{\alpha}(\mu) =
\frac{1}{1-\alpha}\log \bigl( \int f^{\alpha}\, d\lambda\bigr)$ of
order $\alpha\neq 1$. It is easy to check that
(\ref{ref_equ_erdmuarf1}) can be rewritten as 
\begin{equation}
\label{eq:unified}
\lim_{R \rightarrow \infty } e^{rR} D_{\mu }^{\alpha }(R) 
 =   C(r) e^{r h^{a_1}(\mu)}.
\end{equation} 
Setting  $a_1 = \lim_{\alpha\to -\infty} \frac{1- \alpha +
  \alpha r}{1- \alpha + r}  =1-r$ for  $\alpha=\!-\infty$, we also obtain 
 (\ref{sharp_asymp_alph_neg_inf}) from the above expression. 
Also, for $\alpha=0$ we have $a_1=\frac{1}{1+r}$, and 
  (\ref{eq:unified}) reduces to  (\ref{main_asymp}); while for
  $\alpha=1$, we have $a_1=1$, and we formally get back 
(\ref{main_asympalpha1}) since $\lim_{a_1\to 1} h^{a_1}(\mu)
=h^1(\mu) = -\int f\log f\, d\lambda$ (cf.\ Section~\ref{secentr}).
Thus (\ref{eq:unified}) expresses the old and the new asymptotic
  results in a unified form. 

\noindent (c)  Since $\mathrm{ess } \inf \nolimits_{M} f>0$
  the right hand side of (\ref{sharp_asymp_alph_neg_inf}) is finite.
  For the same reason, the right hand side of
  (\ref{ref_equ_erdmuarf1}) is finite for all $\alpha < 0$.  For
  $\alpha \in [ 0,1 )$ the right hand side of
    (\ref{ref_equ_erdmuarf1}) can be shown to be finite by an
    application of H\"older's inequality as in \cite[Remark 6.3
      (a)]{GrLu00}.

\noindent (d) The weak unimodality and continuity conditions on $f$
are the results of our approximation techniques in proving lower
bounds and are probably not necessary. In fact, with a a little
tweaking of the companding approach in the next section one can show
that the right hand sides of (\ref{ref_equ_erdmuarf1}) and
(\ref{sharp_asymp_alph_neg_inf}) still upper bound the asymptotic
performance if these conditions are dropped.

\noindent (e) Note that condition (ii) implies (i). Also, the right
hand side of (\ref{ref_equ_erdmuarf1}) converges to the right hand
side of (\ref{sharp_asymp_alph_neg_inf}) as $\alpha\to -\infty$. 

\noindent (f) The condition $r>1$ is needed in the proof of
the lower bounds on $D_{\mu }^{\alpha }(R)$ where
\cite[Thm.\ 3.1]{Kre10a} is invoked (see
Proposition~\ref{ref_prop_unit_cube}). The upper bounds only need $r\ge 1$
(see Section~\ref{sec_comp}). 

\end{remark}

\section{Distortion and R\'enyi entropy asymptotics of companding quantizers}
\label{sec_comp}

\subsection{Companding quantizers}

Let $N \geq 2$ and $Q_N \in \mathcal{Q}$ denote the $N$-level uniform scalar quantizer with step size
$1/N$ for sources supported in the unit interval $[0,1]$ defined by
$Q_N(x)= 1/2N$ if $x\in [0,1/N]$ and 
\begin{equation}
\label{QNdef}
Q_N(x) = \frac{i-1}{N}+\frac{1}{2N} \quad \text{if\ \ } x\in
  \biggl(\frac{i-1}{N}, \frac{i}{N}\biggr], \quad i=2,\ldots,N.
\end{equation}
The {\em compressor} $G$ 
derived from  a probability density  $g$ on the real line  is the
function
\begin{equation}
\label{def_compress}
G(x) =\int_{-\infty}^x g (y)\,  d \lambda(y).
\end{equation}
Thus the increasing function $G: \mathbb{R} \rightarrow [0,1]$ is the
cumulative distribution function associated with the density $g$.  The
generalized inverse $\hat{G}$ of $G$ is defined by
\[
\hat{G}(y) := \sup\{x: \, G(x)\leq y\}=  \max\{x: \, G(x)\leq y\}
\]
for $y\in (0,1)$. Note that if $g$ is positive almost everywhere with
respect to $\lambda$ (a.e.\ for short), then $G$ is strictly increasing
and $\hat{G}$ is its (ordinary) inverse.

In this paper we will work only with compressor densities $g$ having
compact support, i.e., if $\nu$ denotes the measure induced by $g$, then
$\supp(\nu)$ is  bounded. Thus we can extend the definition of
$\hat{G}$ onto $[0,1]$ by letting  
\[
\hat{G}(0):= \min\{ \supp(\nu)\} > - \infty \quad \text{and}\quad  \hat{G}(1):= \max\{
\supp(\nu)\} < \infty .
\]
The $N$-level {\em companding} quantizer $Q_{g,N}$ associated with
$g$ is defined on $[\hat{G}(0),\hat{G}(1)]$ by 
\[
Q_{g,N}(x) = \hat{G}(Q_N(G(x))).
\]
Note that the codecells of $Q_{g,N}$ are $N$ intervals
$I_{1,N},\ldots,I_{N,N}$ with $I_{1,N}=[\hat{G}(0), \hat{G}(1/N)]$ and 
\[
I_{i,N}=(\hat{G}((i-1)/N), \hat{G}(i/N)], \quad i=2,\ldots,N.
\]
The corresponding quantization points are $\hat{G}((2i-1)/2N),i=1,\ldots,N$.

\begin{remark} (a)  The function $g$ is often called the 
point density for $Q_{g,N}(x)$ since it has the property that for
any $a<b$, 
\[
\lim_{N\to \infty} \frac{1}{N} 
\card ( Q_{g,N}((a,b)) )
= \int_{a}^b g(x)\, d\lambda(x).
\]
(b) If $P_N$ is an arbitrary $N$-level quantizer on $\mathbb{R}$
having convex (interval) codecells, then it can be implemented as a
companding quantizer. In particular, there exists a positive point
density $g$ such that $P_N(x)= Q_{g,N}(x)$ for all (except perhaps a
finite number of) $x \in \mathbb{R}$ (any $x$ such that $ P_N(x)\neq
Q_{g,N}(x)$ is a cell boundary for both quantizers).
\end{remark}

The following result represents the error asymptotics of the compander
if the number of output levels increases without bound. The result
originates with Bennett \cite{Ben48} for $r=2$ and has appeared in
the literature in several different forms (but most often without
precise conditions and a  rigorous proof); see \cite{GrNe98} for a
historical overview.  The proof is given in \ref{appA} and follows the
development in \cite{Lin91} which gives a rigorous
proof for the limit (\ref{bennet}) under different conditions that
include the continuity of $g$ and certain tail conditions, but allow
$f$ and $g$ to have unbounded support. 

\begin{proposition}
\label{bennet_int}
Let $X$ be a random variable with distribution $\mu$ which is
absolutely continuous with respect to $\lambda$ and let $f$ denote its
density.  Let $G$ be a compressor with point density $g$.  Assume that
the support of $\mu$ is included in a compact interval $I$ such that
$\mathrm{ess } \inf_{I} g > 0$ and $g(x)=0$ a.e.\ on $\mathbb{R}
\setminus I$.  Then for $r\ge 1$,
\begin{equation}
\label{bennet}
\lim_{N \rightarrow \infty } N^{r} D_{\mu }(Q_{g,N}) = C(r) 
\int_I  \frac{f}{g^{r}} \, d \lambda.
\end{equation}
\end{proposition}

\begin{remark}
\label{comp_dens_int}
Since  $\mathrm{ess } \inf_{I} g > 0$
we know that $\mu$ is absolutely continuous with respect to
$g \lambda$ and $\int_{ I } \frac{f}{g^{r}} \, d \lambda <
\infty$. 
\end{remark}

\subsection{R\'enyi entropy asymptotics of companding quantizers}
\label{secentr}

In order to be able to construct asymptotically optimal companding quantizers
(cf.\ (\ref{def_asymp_opt_quant})), in addition to the asymptotic
distortion,  we also have to
control the  quantizer's entropy, at least for
high rates. In this section we derive a result (Proposition~\ref{ref_rel_renyi_entr}) which asymptotically describes the R\'enyi
entropy of the compander as a function of the number of quantization
points. Let  $1_{A}$ denote the indicator function of $A\subset \R$.

\begin{definition}
\label{ref_Def_renyi_entr_diff}
Let $\mu$ be absolutely continuous with respect to $\lambda$ with density $f$
 and define  $M = (\inf ( \supp (\mu ) ), \sup ( \supp (\mu ) ))$.
Let $\alpha \in [-\infty, \infty ]$ and assume that 
\begin{itemize}
\item[(i)]
$1_{ \supp (\mu ) } f^{\alpha}$ is integrable if $\alpha \in {} (-\infty, \infty ) \setminus \{ 1 \} $,
\item[(ii)]
$\mathrm{ess } \inf _{M} f > 0$
if $\alpha = - \infty$,
\item[(iii)]
$f \log f $ is integrable if $\alpha = 1$,
\item[(iv)]
$\mathrm{ess } \sup _{\mathbb{R}} f < \infty$ if $\alpha = + \infty$.
\end{itemize}
Then the  R\'enyi differential entropy of order $\alpha$ of  $\mu$ is defined by
\[
h^{\alpha }(\mu )=
\begin{cases}
\frac{1}{1 - \alpha } \log ( \int_{ \supp (\mu ) } f^{\alpha } \, d
  \lambda ),  &\alpha \in {} (-\infty, \infty ) \setminus \{ 1 \} \\ 
- \int f \log f \, d \lambda, &\alpha = 1 \\  
- \log ( \mathrm{ess } \sup _{\mathbb{R}} f ), &\alpha =  \infty\\
- \log ( \mathrm{ess } \inf _{M} f),
&  \alpha = - \infty.
\end{cases}
\]
\end{definition}

\begin{remark}
Just as in  the case of R\'enyi entropy  (see Remark
\ref{remark_hospital}) the mapping
$[-\infty, \infty] \ni \alpha \rightarrow h^{\alpha}(\mu )$ is
continuous for the differential entropy. 
\end{remark}

Recall that $\hat{q}_{\Delta}$ denotes the infinite-level uniform quantizer
with step-size $\Delta>0$. Recall $M$ from Definition \ref{ref_Def_renyi_entr_diff}  and let
$A({\Delta},M)=\{ a \in \hat{q}_{\Delta }(\mathbb{R}) : \hat{q}_{\Delta}^{-1}(a)
\subset M \}$ and
\[
q_{\Delta, M}(\cdot ) = \sum_{a \in A({\Delta },M)} a \cdot
1_{\hat{q}_{\Delta }^{-1}(a)}(\cdot ) . 
\]

The following result is due to R\'enyi \cite{Ren60a} and
Csisz\'ar \cite{Csi73} for $\alpha\in (0,\infty)$. The 
proof for $\alpha\in [-\infty,0]$ is given in \ref{appA}.

\begin{lemma}
\label{ref_lemma_xyz}
Let $\mu$ be absolutely continuous with respect to $\lambda$ having
density $f$.  Let $\alpha \in [ -\infty,
\infty )$ and assume that the R\'enyi differential entropy of order
$\alpha$ of 
$\mu$ exists and is finite.  Assume that $H_{\mu }^{\alpha
}(\hat{q}_{\Delta }) < \infty$ 
for some $\Delta > 0$.  If $\alpha \in {} (- \infty, \infty )$, then
\[
\lim_{\Delta \rightarrow 0} \left(  H_{\mu }^{\alpha }(\hat{q}_{\Delta }) + \log ( \Delta ) \right) =
h^{\alpha }(\mu ).
\]
Moreover,
\[
\lim_{\Delta \rightarrow 0} \left(  H_{\mu }^{- \infty}(\hat{q}_{\Delta , M}) + \log ( \Delta ) \right) =
h^{-\infty }(\mu ).
\]
\end{lemma}

Next we define the R\'enyi relative entropy between two probability
measures for the case where both have densities. 

\begin{definition}
\label{def_relent}
Let $\mu$ and $\nu$ be probability measures which are absolutely continuous with
respect to $\lambda$. Denote by $f$ and $g$ the densities 
of $\mu$ and $\nu$. 
Moreover, assume that $\mu$ is absolutely continuous with respect to $\nu$ and, therefore,
we assume w.l.o.g.\ that $\{ g = 0 \} \subset \{ f = 0 \}$.
Setting
\[
E = \{f>0 \} \quad \text{ and } \quad M = {} ( \inf (\supp (\mu )), \sup (\supp (\mu )) )
\]
the R\'enyi relative entropy of order $\alpha$ between the
distributions $\mu$ and $\nu$ is defined as
\begin{equation}
\label{def_rel_renyi_entropy}
\mathcal{D}_{\alpha}( \mu \| \nu ) = 
\begin{cases}
\frac{1}{\alpha-1} \log\biggl( \int_{E}
f^{\alpha } g^{1-\alpha} \, d\lambda \biggr), 
 &\alpha \in {} (-\infty, \infty) \setminus \{ 1 \} \\ 
\int_{E} f \log \frac{f}{g}\, d\lambda, &\alpha = 1 \\
\log ( \mathrm{ess } \sup\nolimits_{E} \frac{f}{g} ), &\alpha = \infty \\
\log ( \mathrm{ess } \inf\nolimits_{M} \frac{f}{g} ), &\alpha = -
        \infty.
\end{cases}
\end{equation}
(For $\mathcal{D}_{-\infty}(\mu\|\nu)$ to be well defined, we need the condition
$ \mathrm{ess } \inf\nolimits_{M} f>0$.)
\end{definition}

The following result determines the asymptotics of the
R\'enyi entropy of a companding quantizer.

\begin{proposition}
\label{ref_rel_renyi_entr}
Let $\alpha \in [-\infty , \infty )$.
Suppose   $\mu$ and $\nu$ are as in Definition~\ref{def_relent} and
$\mathcal{D}_{\alpha}(\mu \| \nu )<\infty$. Then  
\[
\lim_{N\to \infty}\biggl(  H_{\mu}^{\alpha}(Q_{g, N})-  \log N\biggr) =
-\mathcal{D}_{\alpha}(\mu \| \nu ).
\] 
\end{proposition}

\begin{remark} (a) For the sake of distortion analysis we previously specified
that $g$ has bounded support, but  in this proposition the only condition on
$f$ and $g$ is the finiteness of
$\mathcal{D}_{\alpha}(\mu\|\nu)$. 

\noindent{(b)} In a sense, the proposition generalizes 
Lemma~\ref{ref_lemma_xyz}. Indeed, if the support of $\mu$ is included
in a compact interval $I$ and $g$ is the uniform density on $I$, then
$Q_{g,N}$ is the uniform quantizer of step-size $\Delta_N=\lambda(I)/N$
over $I$, and the proposition reduces to Lemma~\ref{ref_lemma_xyz}
(for the sequence of step-sizes $\Delta_N$).
\end{remark}

\begin{proof}
Recall the definition of the compressor $G$ from (\ref{def_compress}).
We proceed in two steps. \smallskip \\
\emph{1.} \   We  show that $h^{\alpha }(\mu \circ G^{-1}) = -\mathcal{D}_{\alpha}(\mu \| \nu )$
for every $\alpha \in [- \infty, \infty)$.
\smallskip 

Let $\alpha \in {} (-\infty, \infty ) \setminus \{ 1 \}$ and let $f_G$
be the density of $\mu \circ G^{-1}$ (see Lemma~\ref{dens_inv_compr}
in \ref{appA}).
Definition \ref{ref_Def_renyi_entr_diff} and Lemma~\ref{dens_inv_compr} imply
\begin{eqnarray*}
h^{\alpha }(\mu \circ G^{-1}) & = &  
\frac{1}{1- \alpha } \log  \int (f_{G})^{\alpha }\, d \lambda  \\
&=& 
\frac{1}{1- \alpha } \log  \int (f(\hat{G}(y)) \hat{G}^{\prime }(y))^{\alpha }\, d \lambda(y)  \\
&=&
\frac{1}{1-\alpha}\log\biggl( \int_{\hat{G}^{-1}(E)}
f(\hat{G}(y))^{\alpha} g(\hat{G}(y))^{1-\alpha} \hat{G}^{\prime}(y)\, d \lambda(y) \biggr)  \\  
&=&   \frac{1}{1-\alpha}\log\biggl( \int_{E}
f(x)^{\alpha} g(x)^{1-\alpha} \, d \lambda(x) \biggr)  =
-\mathcal{D}_{\alpha}(\mu \| \nu) 
\end{eqnarray*}
where in the penultimate equality we used again the chain rule for the
Lebesgue integral (see \cite[Corollary 4]{SeVa69}), which is applicable due to
the monotonicity of $\hat{G}$ and the integrability of
$f^{\alpha}g^{1-\alpha}$ (which follows from the finiteness of 
$\mathcal{D}_{\alpha}(\mu \| \nu)$). Note that the above chain of
equalities implies that  $(f_G)^{\alpha}$ is integrable. One can
deduce the assertion of step 1  
for $\alpha \in \{ -\infty, 1 \}$ in a very similar manner.
\smallskip \\ 
\emph{2.} \  Now we prove the assertion of the proposition. Since
$G$ is increasing and continuous, $\hat{G}$ is strictly increasing on
$(0,1)$. Recall the definition of $Q_N$ in (\ref{QNdef}) and note that
$Q_N=\hat{q}_{1/N}$ on $(0,1)$. Then 
\begin{equation}
\label{entropy_ident}
H_{\mu}^{\alpha }(Q_{g,N}) = 
H_{\mu \circ G^{-1}}^{\alpha } ( Q_{N} ) 
\end{equation}
for all $\alpha \in [- \infty, \infty)$.  Since $\mu \circ
G^{-1}((0,1))= 1$, we obtain $H_{\mu \circ G^{-1}}^{\alpha } ( Q_{N} )
= H_{\mu \circ G^{-1}}^{\alpha } ( \hat{q}_{1/N} )$. In view of
(\ref{entropy_ident}) we deduce $H_{\mu}^{\alpha }(Q_{g,N}) = H_{\mu
  \circ G^{-1}}^{\alpha } ( \hat{q}_{1/N} )$.  From step 1 and by the
assumption we know that $h^{\alpha }(\mu \circ G^{-1})$ is finite.
Since $\hat{q}_{1/N}$ has no more than $N$ cells with nonzero $\mu \circ
G^{-1}-$measure, the entropy $H_{\mu \circ G^{-1}}^{\alpha}(\hat{q}_{1/N})$
is also always finite.  Lemma~\ref{ref_lemma_xyz} and step 1 imply
\begin{eqnarray*}
\lim_{N\to \infty}\biggl(H_{\mu }^{\alpha}(Q_{g,N})-  \log N\biggr) &=& 
\lim_{N\to \infty}\biggl(H_{\mu  \circ G^{-1}}^{\alpha}(\hat{q}_{1/N}) - \log N \biggr) \\
&=& h^{\alpha } (\mu \circ G^{-1}) = -\mathcal{D}_{\alpha}(\mu \| \nu).
\end{eqnarray*}
\end{proof}

\begin{remark}
Although we do not need this fact in the sequel it is worth noting 
that Lemma~\ref{ref_lemma_xyz} and Proposition~\ref{ref_rel_renyi_entr} 
are also valid for $\alpha = \infty$. For example, by  an application
of Lebesgue's density theorem one can show that 
\[
\lim_{\Delta \rightarrow 0} 
\frac{\sup \{ \mu(\hat{q}_{\Delta }^{-1}(a)) : a \in \hat{q}_{\Delta}(\mathbb{R})  \} }{\Delta }
= {\mathrm{ess } \sup}_{\mathbb{R}} f,
\]
which yields the assertion of Lemma~\ref{ref_lemma_xyz} for $\alpha = \infty$. 
Generalizing the proof of  Proposition~\ref{ref_rel_renyi_entr}  to
$\alpha = \infty$ is straightforward.

\end{remark}

\subsection{Optimal point densities}

Combining the previous results  we can 
find a companding quantizer  which provides an (asymptotic) upper bound for the
optimal quantization error. Later on we will show that this quantizer is an
asymptotically optimal one. Recall  definition (\ref{lambda12def}) of
$a_1$ and $a_2$.

\begin{corollary}
\label{asymp_quant_comp}
Let $r\ge 1$ and $\alpha \in [-\infty, 1+r )$. Assume that $\mu$ is
  supported on a compact interval $I$ and has  density
  $f$ such that  $\mathrm{ess } \inf_I f > 0$. Moreover, assume that
  $f^{a_1}$ is integrable if $\alpha \in {} (1, 1+r) $ and $f
  \log f$ is integrable if $\alpha = 1$.  Let
\begin{equation}
\label{def_opt_dens}
f^{*} 
=
\begin{cases}
( \int_I f^{1 / a_2}\, d \lambda )^{-1} f^{1 / a_2}, 
 &\alpha \in {} (-\infty, 1+r) \setminus \{ 1 \} \\ 
(\lambda(I))^{-1}1_{I}, &\alpha = 1  \\
f, &\alpha = - \infty.
\end{cases}
\end{equation}
Then,
\begin{equation}  
\label{eq:comp_asymp}
\lim_{N \rightarrow \infty } e^{r H_{\mu }^{\alpha }(Q_{f^{*},N})} D_{\mu }(Q_{f^{*},N})  
= 
\begin{cases}
C(r) (\int_I f^{a_1 }\, d \lambda
  )^{a_2},  & \alpha \in (-\infty, 1+r) \setminus \{ 1 \} \\ 
 C(r) e^{-r \int f \log f\,  d \lambda}, &\alpha = 1\\
C(r) \int_I f^{1-r}\, d \lambda, &\alpha = - \infty.
\end{cases}
\end{equation}
\end{corollary}
\begin{proof}
It is not hard to show using H\"older's inequality that $f^{a_1}1_I$
is integrable for every $\alpha \in {} (-\infty, 1+r) \setminus \{ 1
\}$ (cf.\ \cite[Remark 6.3 (a)]{GrLu00}).  Moreover $f^{1-r}$ is
integrable.  Clearly, $f^{1 / a_2} 1_{I}$ is integrable for every
$\alpha \in {} (-\infty, 1+r)$. These facts imply that $f^*$ is well
defined (note that $\mathrm{ess } \inf_I f^* > 0$), $\int f/(f^*)^r \,
d\lambda<\infty$, and the integrals on the right hand side of
(\ref{eq:comp_asymp}) are finite.  It is also easy to check that $
\mathcal{D}_{\alpha }(\mu \| f^{*} \lambda )$ is finite. Thus we can
apply Propositions~\ref{ref_rel_renyi_entr} and \ref{bennet_int}. We
obtain
\begin{eqnarray*}
 \lim_{N \rightarrow \infty } e^{r H_{\mu }^{\alpha }(Q_{f^{*},N})} D_{\mu }(Q_{f^{*},N}) 
&=& 
\lim_{N \rightarrow \infty } e^{- r \mathcal{D}_{\alpha }(\mu \| f^{*} \lambda )} N^{r} D_{\mu }(Q_{f^{*},N}) 
\nonumber \\
&=& 
e^{- r \mathcal{D}_{\alpha }(\mu \| f^{*} \lambda )} C(r) \int_I
\frac{f}{(f^{*})^{r}}\,  d \lambda.
\end{eqnarray*}
Now (\ref{def_opt_dens}) and (\ref{def_rel_renyi_entropy}) yield the assertion.
\end{proof}

\begin{remark} For $\alpha\in [-\infty,1+r)$ the point density $g=f^*$
  in the corollary minimizes the asymptotic 
  performance $\lim\limits_{N\to \infty} e^{r H_{\mu }^{\alpha
    }(Q_{g,N})} D_{\mu }(Q_{g,N})$.  For $\alpha=0$ and
  $\alpha=1$ this optimal choice of $g$ has long been
  known. In the  case 
   $\alpha\in (-\infty,1+r)\setminus\{1\}$, 
  by Propositions \ref{bennet_int} and \ref{ref_rel_renyi_entr} the
  above limit is proportional to 
\[
\left( \int_I f^{\alpha } g^{1-\alpha} \, d \lambda
\right)^{\frac{r}{1-\alpha }} \int_I f g^{-r}\,  d \lambda
\]
and H\"older's inequality (for $\alpha<1$) or the reverse H\"older
inequality (for $\alpha\in (1,1+r)$) can be used to show that this
functional is minimized by $g=f^*$. The resulting minimum is $\bigl(\int_I
f^{a_1 }\, d \lambda \bigr)^{a_2}$. The case $\alpha=-\infty$ follows by
letting $\alpha\to -\infty$.
\end{remark}

\section{Some important properties of optimal scalar quantization}
\label{sec_fund}

Define
\[
i(f) = \mathrm{ess } \inf\nolimits_{\supp (\mu )} f , \qquad 
s(f) = \mathrm{ess } \sup\nolimits_{\supp (\mu )} f .
\]
For the case  $\alpha =0$ the following result is originally due to Pierce (\cite{Pie70}, 
\cite[Lemma 6.6]{GrLu00}). In our proof, given in \ref{appB},  we use
a refined version  
provided by Luschgy and Pag\`es \cite[Lemma 1]{LuPa08}.

\begin{proposition}
\label{ref_lemm_pierce}
(i) If $R \geq 1$ and $\int | x | ^{r+\beta} \, d \mu (x) < \infty$
for some $\beta > 0$,  
then there exists a constant $C_0 > 0$  (which depends only on $r$ and
$\beta$) such that 
\[
e^{rR} D_{\mu }^{\alpha }(R) \leq   
C_0 \left( \int  | x | ^{r+\beta }\,  d\mu (x)  \right)^{r/(r + \beta )} 
\]
for every $\alpha \in [0, \infty ]$. \\
(ii) Suppose $\supp(\mu)$  is a compact interval and  $\mu$
absolutely continuous with respect to $\lambda$ 
with density  $f$. Assume that $i(f) > 0$.
Then for all  $\alpha < 0$
\begin{equation}
\label{upp_bou_neg_alp}
e^{rR}D_{\mu}^{\alpha }(R) \leq \frac{2^{r}}{i(f)^{r}}.
\end{equation}

\end{proposition}

As an immediate consequence we obtain the following.
\begin{corollary}
\label{ref_coro_lim0}
Under either condition (i) or (ii) of
Proposition~\ref{ref_lemm_pierce} we have 
$\lim_{R \rightarrow \infty } D_{\mu }^{\alpha }(R) = 0$.
\end{corollary}

Let $\diam (A) = \sup \{ | x - y | : x,y \in A \}$ denote the diameter
of an arbitrary non-empty set $A \subset \mathbb{R}$. The next result
shows that the measure of the codecells of optimal quantizers tends to
zero for absolutely continuous distributions.  The proof, given in
\ref{appB}, adopts some techniques of Gray \emph{et al$.$} \cite[Proof
of Lemma 11]{GrLiLi02}.

\begin{lemma}
\label{ref_lemm_gray}
Let $\mu$ be absolutely continuous 
with respect to 
$\lambda $ having density $f$.
Assume further either of the following conditions
\begin{itemize}
\item[(i)] $\alpha \in [0, \infty]$ and $\int | x | ^{r+\beta}\, d \mu
  (x) < \infty$ for some $\beta > 0$,
\item[(ii)] $\alpha < 0$ and $\supp(\mu)$ is a compact 
interval  and $0<i(f) \leq s(f) < \infty$. 
\end{itemize}
Then for every $\varepsilon > 0$ there exists an  $R_{0}>0$ 
with the property that for every $R \geq R_{0}$ 
there is a  $\delta > 0$ such that
\begin{equation}
\label{ref_iequ_b_01}
\max \{ \mu ( q^{-1}(a) ) : a \in q(\mathbb{R})  \}  < \varepsilon 
\end{equation}
for every $q \in \mathcal{Q}$ with $H_{\mu}^{\alpha }(q) \leq R$ and
$| D_{\mu }(q) - D_{\mu }^{\alpha }(R) | < \delta$.
If, additionally, in case (i) the support of $\mu$ consists of $m \geq 1$ compact intervals $I_{1},\ldots,I_{m}$ 
and $i(f) > 0$, then in both cases (i) and (ii) we have
\begin{equation}
\label{ref_equ_onedim} 
\max \{ \diam ( q^{-1}(a) \cap I_{i} ) :  a \in q(\mathbb{R}), i \in \{ 1,\ldots,m \} \} < \varepsilon \cdot i(f)^{-1}
\end{equation}
where $m=1$ and $I_1=I$ for case (ii).
\end{lemma}

Let $\mu$ be absolutely continuous with respect to $\lambda$
and denote the  density of $\mu$ with $f$.
Let 
\begin{equation}
\label{const_quant}
C=\left( \frac{i(f)}{s(f)} \right)^{\frac{r+1}{r}}
\left(\frac{1}{4^{r}(1+r)}\right )^{1/r} \in (0, 1) . 
\end{equation}
For any $q \in \mathcal{Q}$ let 
\[
N_{q} = \{ a \in q(\mathbb{R}) : \mu (q^{-1}(a)) > 0 \} .
\]
In the case  $\alpha < 0$ we need to control in our proofs
the cardinality of the codebook of any quantizer whose entropy is
less than or equal to the rate constraint $R$.  To this end, for $R 
\geq 0$, we define 
\[
\mathcal{H}_{R} = 
\{ q \in \mathcal{Q}^c:  \quad H_{\mu }^{\alpha }(q) \leq R, \quad 
C e^{R} \leq \card ( N_{q} ) \leq e^{R}  \} .
\]
In addition, we will have to control the difference between the rate
constraint  and the entropy of the quantizer.
Thus, for $\alpha \in (- \infty , 0)$, arbitrary constant $\kappa>0$, 
and $R> \log ( \frac{2^{1-\alpha }-1}{\kappa} )$, we define
\[
\mathcal{K}_{R} = \mathcal{K}_{R}(\kappa)=
\left\{ q \in \mathcal{H}_{R} :
e^{R - H_{\mu }^{\alpha }(q)} \leq \biggl( \frac{1}{1-(2^{1-\alpha
  }-1)\kappa^{-1}e^{-R}} \biggr)^{1/(1-\alpha )} \right\}. 
\]
The next lemma is proved in \ref{appB}.

\begin{lemma}
\label{pierce_neg_para}
Let $\mu$ be absolutely continuous 
with respect to 
$\lambda $ having density $f$. 
Assume that $\supp (\mu )$  is a  compact interval and  $0 < i(f)  \leq s(f)  < \infty .$
For every $\alpha \in [-\infty, 0]$ and $R \geq 0$ we have 
\begin{equation}
\label{dmuealphrhr}
D_{\mu }^{\alpha }(R) = \inf \{ D_{\mu }(q) : q \in \mathcal{H}_{R} \}.
\end{equation}
If $\alpha \in (- \infty , 0)$ and $R> \log ( \frac{2^{1-\alpha }-1}{C} )$, then 
\begin{equation}
\label{dmuealphrkr}
D_{\mu }^{\alpha }(R) = \inf \{ D_{\mu }(q) : q \in \mathcal{K}_{R}(C) \}.
\end{equation}
\end{lemma}

We let $U(I)$ denote the uniform distribution on a bounded interval $I
\subset \mathbb{R}$ with positive length.  Let $m \geq 2$ and let
$I_{1},\ldots,I_{m}$ be a partition of $I$ into $m$ intervals of equal
length $\diam (I) /m$. Let $s_{1},\ldots,s_{m} \in {} (0,1)^{m}$ with
$\sum_{i=1}^{m}s_{i}=1$ and assume the source distribution is of the
form $\mu = \sum_{i=1}^{m} s_{i} U(I_{i})$.  Of special interest in
our proofs are the codecells which are  straddling the intervals $I_{i}$.
Hence we define 
for any quantizer $q \in \mathcal{Q}$ the sets 
\begin{equation}
\label{aqsqdef}
A(q) = \bigcup_{i=1}^m  \{ a \in q( \mathbb{R} ) :
\lambda (q^{-1}(a)\setminus  I_{i})=0 \} \quad \text{and} \quad S(q)=q(\mathbb{R}) \setminus A(q).
\end{equation}
In the proof of our main result we have to ensure that the
contribution of the straddling cells to the overall entropy of the
quantizer can be (asymptotically) neglected.  For $\alpha < 0$ this is
the case if it suffices to consider only quantizers with the property
that the length of each straddling cell is at least as large as a certain
(fixed)  constant times  the length of the smallest non-straddling cell.
Exactly this is ensured by the following lemma which sharpens
Lemma~\ref{pierce_neg_para}. The proof is given in  \ref{appB}.
Recall the definition (\ref{const_quant}) of the constant  $C$ and let
$\kappa\in (0,C)$. For $R>\log ( \frac{2^{1-\alpha }-1}{\kappa})$ let
\begin{eqnarray}
\mathcal{G}_R  = \mathcal{G}_R  (\kappa)
&=& \bigl\{ q \in \mathcal{K}_{R}(\kappa) :   
2 \inf \{ \diam ( q^{-1}(a)\cap I ) : a \in S(q) \} \nonumber \\
&& \qquad \geq  
\inf \{ \diam ( q^{-1}(a) ) : a \in A(q) \} \bigr\}.
\nonumber
\end{eqnarray}

\begin{lemma}  
\label{lemm_inf_g}
Assume that $\mu = \sum_{i=1}^{m} s_{i} U(I_{i})$ is a
piecewise uniform distribution as specified above.  Let  $r>1$  and
$\kappa \in (0,C)$. Then for every $\alpha \in (-\infty, 0)$ there is an
$R_{0}(\kappa)>0$ such that for every $R \geq R_{0}(\kappa)$,
\[
D_{\mu }^{\alpha }(R) = \inf \{ D_{\mu }(q) : q \in \mathcal{G}_R(\kappa) \} .
\] 
\end{lemma}

A bijective mapping $T: \mathbb{R} \rightarrow \mathbb{R}$ is called 
a similarity transformation if there exists $c \in (0, \infty )$, the
scaling number, such 
that $|Tx-Ty| = c |x-y|$ for every $x,y \in \mathbb{R}$. 
The last result of this section 
describes how the optimal
quantization error scales under a similarity transformation.
For $\alpha = 0$ the reader is also referred to \cite[Lemma
  3.2]{GrLu00}. 
Let us denote by
\[
C_{\mu}^{\alpha }(R) = \{ q \in \mathcal{Q}^{c} : D_{\mu }(q) = D_{\mu }^{\alpha }(R) \}
\]
the set of all optimal quantizers in $\mathcal{Q}^{c}$ for $\mu$ under R\'enyi entropy
constraint $R$ of order $\alpha$.

\begin{lemma}
\label{ref_lemm_scale}
Let $\alpha \in [-\infty, \infty]$ and $T: \mathbb{R} \to \mathbb{R}$ be a similarity
transformation with scaling number $c > 0$. Then for any $R \geq 0$ we have
\[
D_{\mu \circ T^{-1}}^{\alpha }(R) = c^{r} D_{\mu }^{\alpha }(R) .
\]
Moreover, 
\[
C_{\mu \circ T^{-1}}^{\alpha }(R) = \{ T \circ q \circ T^{-1} : q \in C_{\mu }^{\alpha }(R)  \}.  
\]
\end{lemma}

\begin{proof} The lemma follows because for any $q\in
  \mathcal{Q}$   we have $\bar{q}:=
  T\circ q \circ T^{-1}\in  \mathcal{Q}$, 
 $ H_{\mu \circ T^{-1} }^{\alpha}(q) = H_{\mu }^{\alpha}(\bar{q})$, and
  $ D_{\mu \circ T^{-1} }(q) = c^rD_{\mu }(\bar{q})$ (also, $q\in
  \mathcal{Q}^c$  iff $\bar{q}\in  \mathcal{Q}^c$). See  also \cite[Lemma 2.4]{Kre10a} where  $\alpha\ge 0$ and $r>1$
are considered, but the same proof clearly works for all 
$\alpha < 0$ and $r>0$. 
\end{proof}

\section{Inequalities for mixture distributions}
\label{sec_mix}

In this section  we provide  upper and lower bounds for the 
optimal quantization error of mixture distributions in terms of the
optimal quantizer performance for the component distributions. Proofs are given in \ref{appC}. 

\begin{definition}
\label{ref_def_mdivis}
Let $m \geq 2$ and 
$A_{1},\ldots,A_{m}$ be measurable sets which are pairwise disjoint.
The distribution $\mu$ is called $m-$divisible with respect to 
$(A_{1},\ldots,A_{m})$  if $\mu ( A_{i} ) > 0$ for all $i=1,\ldots,m$ and
$\mu ( \cup_{i=1}^{m} A_{i} )=1$.
\end{definition}

For any measurable $A \subset \mathbb{R}$ with $\mu (A) > 0$ we let 
$\mu ( \cdot | A )$ denote  the conditional probability of $\mu$ with
respect to  $A$, i.e., $\mu(B|A)=\mu(B\cap A)/\mu(A)$ for all
measurable $B\subset \R$. 
If $\mu$ is $m-$divisible, then we write $\mu_{i}=\mu ( \cdot | A_{i} )$.

\begin{proposition}
\label{ref_pro_zerleg}
Let $R \geq 0$,  $\alpha \in [0,\infty ) \setminus \{ 1 \}$, and $m \geq 2$. 
Assume that $\mu$ is $m-$divisible with partition
$(A_{1},\ldots,A_{m})$. Moreover assume, that  
$\int | x | ^{r} d \mu_{i}(x) < \infty $ for every $i=1,\ldots,m$.
Let $R_{1},\ldots,R_{m} \in [0,\infty )$. Letting $s_{i}=\mu ( A_{i} )$, 
we have 
\[
D_{\mu }^{\alpha } (R) \leq \sum_{i=1}^{m} s_{i} D_{\mu_{i} }^{\alpha } \left( R_{i} \right)
\]
if either one of 
the following inequalities holds:
\begin{eqnarray}
\label{ref_cond_01}
\log \left( \sum_{i=1}^{m} s_{i}^{\alpha} e^{(1-\alpha )R_{i}} \right) & \leq & (1 - \alpha )R 
\qquad \text {if } \alpha \in [0,1 ) , \\
\label{ref_cond_02}
\log \left( \sum_{i=1}^{m} s_{i}^{\alpha} e^{(1-\alpha )R_{i}} \right) & \geq & (1 - \alpha )R 
\qquad \text {if } \alpha \in ( 1,\infty ).
\end{eqnarray}
\end{proposition}

Recall the definition (\ref{lambda12def}) of $a_1$ and
$a_2$. 
Let $m \geq 2$ and $s_{1},\ldots,s_{m} \in (0,1)^{m}$ with
$\sum_{i=1}^{m}s_{i}=1$. For every $i \in \{1,\ldots,m\}$ and $\alpha \in
[0,r+1 ) \setminus  \{ 1 \}$ let
\begin{equation}
t_{i} = s_{i}^{1/a_2}  
\left( \sum_{j=1}^{m} s_{j}^{ 
a_1} \right)^{-\frac{1}{1-\alpha }} .
\label{ref_def_ti}
\end{equation}

\begin{lemma}
\label{ref_lemm_smu12}
Let $m \geq 2$. Let $\mu$ be non-atomic and $m-$divisible with respect
to $(A_{1},\ldots,A_{m})$. 
Assume  $\int | x |^{r} \, d \mu _{i} (x) < \infty$ for all $i=1,\ldots,m$. 
Let $i_{0} \in \{1,\ldots,m\}$ with $\mu (A_{i_{0}} )=s=\max \{ \mu (A_{i} ) : i=1,\ldots,m \}$.
If $\alpha \in [0,1) $, then   
\begin{equation}
\label{ref_equ_ergnjks}
\liminf_{R \rightarrow \infty } e^{r R} D_{\mu }^{\alpha }(R) \geq 
s^{ a_1 a_2 } \liminf_{R \rightarrow \infty}
e^{rR} D_{\mu_{i_{0}}}^{\alpha } ( R  ) .
\end{equation}

Let $s_{i}=\mu ( A_{i} )$ 
and assume $\alpha \in [0,r+1 ) \setminus \{ 1 \}$. 
Then we have
\begin{equation}
\label{ref_iequ_rrinferr}
\limsup_{R \rightarrow \infty } e^{rR} D_{\mu }^{\alpha }(R) \leq  
\sum_{i=1}^{m}  s_{i} t_{i} ^{-r} 
\limsup_{R \rightarrow \infty } e^{rR} D_{\mu_{i} }^{\alpha }(R).
\end{equation}

\end{lemma}

\section{Proof of main result}
\label{sec_proof}

Recall that $U(I)$ denotes the uniform distribution on a bounded
interval $I$ with positive length. First we show that the optimal
quantizer performance for $U(I)$ is the same for all negative $\alpha$.

\begin{lemma}
\label{uni_neg_para}
Let $-\infty<a<b<\infty$. For every $R \geq 0$
and $\alpha < 0$, we have  
\[
D_{U([a,b]) }^{\alpha}(R)=D_{U([a,b]) }^{0}(R).
\]
\end{lemma}

\begin{proof}

Note that by Lemma~\ref{ref_lemm_scale} it suffices to consider the
case $[a,b]=[0,1]$.
Since $D_{U([a,b])}^{\alpha }(R)$ is nonincreasing in $\alpha$ by Lemma~\ref{lemm_mono} it suffices to prove the assertion
for $\alpha = - \infty$. Let $R \geq 0$ 
and assume $q \in \mathcal{Q}$ satisfies  $H_{\mu}^{- \infty }(q) \leq R$.
Setting
\[
N_{q} = \{ a \in q(\mathbb{R}) : \mu (q^{-1}(a)) > 0 \}
\]
this condition is equivalent to 
\begin{equation}
\label{ref_unendl}
p = \min \{ \mu (q^{-1}(a)) : a \in N_{q} \} \geq \exp(-R) .
\end{equation}
Let $\lfloor x\rfloor $ denote the largest
integer less than or equal to $x\in \R$. Using $1 \geq \card ( N_{q} ) \cdot p$ we get
\begin{equation}
\label{ref_null}
\card ( N_{q}  ) \leq \lfloor \exp(R) \rfloor ,
\end{equation}
which is equivalent to $H_{\mu}^{0}(q) \leq R$.  From, e.g.,
\cite[Example 5.5]{GrLu00} we know that only the quantizer $g \in
\mathcal{Q}$ which partitions the unit interval into $\lfloor \exp(R)
\rfloor$ intervals of equal length with their midpoints as
quantization points, attains the optimal error, i.e.,
$D_{U([0,1])}(g)=D_{U([0,1])}^{0}(R)$. But this quantizer satisfies
conditions (\ref{ref_unendl}) and (\ref{ref_null})
simultaneously. Hence, $D_{U([0,1])}(g)=D_{U([0,1])}^{-\infty}(R)$,
which yields the assertion.
\end{proof}

Next we determine the exact
behavior of $D_{U([0,1])}^{\alpha }(R)$ for large $R$.  For $\alpha =
1$ the following result is from 
\cite{GyLi00}. For the case $\alpha = 0$ the reader is referred, for
example, to  \cite[Example 5.5]{GrLu00}.

\begin{proposition}
\label{ref_prop_unit_cube}
Let $r > 1$ and $R>0$. Let $-\infty<a<b<\infty$. Then the following hold:
\begin{itemize}
\item[(i)]  If $\alpha \in [0 , r+1)$, an
optimal quantizer always exists for $U([a,b])$, i.e., we can find a $q
\in \mathcal{Q}$ with $H_{U([a,b])}^{\alpha }(q) \leq R$ and
$D_{U([a,b])}(q) = D_{U([a,b])}^{\alpha }(R)$. 

\item[(ii)] Suppose $\alpha \in [0 , r+1)$ and let $n \in \mathbb{N}$
  be such that $R \in ( \log (n) , \log ( n+1 )]$. Then the
  restriction to $[a,b]$ of the quantizer $q$ in (i) has $(n+1)$
  interval cells, $n$ of which are of equal lengths and one having
  length less than or equal to that of the others.  If $\alpha>0$,
  then $q$ meets the entropy constraint with equality, i.e.,
  $H_{U([a,b])}^{\alpha }(q) = R$.

\item[(iii)] For all  $\alpha \in [-\infty , r+1)$, we have 
\begin{equation}
\label{ref_equ_duoe45}
\lim_{R\to \infty} e^{rR} D_{U([0,1])}^{\alpha }(R) = C(r) . 
\end{equation}
\end{itemize}
\end{proposition}
\begin{proof}
Assertions (i) and (ii)
follow directly from \cite[Thm.\ 3.1]{Kre10a} by noting that in view of
Lemma~\ref{ref_lemm_scale} it suffices to consider the
case $[a,b]=[0,1]$

To prove (iii) first we note that by Lemma~\ref{uni_neg_para}, the limit 
(\ref{ref_equ_duoe45}) holds for all $\alpha\in [-\infty,0)$ since it
  holds for $\alpha=0$.  Thus we need only concentrate on the case
  $\alpha \in ( 0, r+1 ) \setminus \{ 1 \}$. Applying
  \cite[Thm.\ 3.1]{Kre10a} we obtain
\[
D_{U([0,1])}^{\alpha }( \log (n) ) = C(r) n^{-r}
\]
for every $n \in \mathbb{N}$.
Now let $R \geq 0$ and $n_{R} \in \mathbb{N}$, such that $\log(n_{R})<
R \le  \log ( n_{R}+1 )$.
We get
\begin{eqnarray*}
\lefteqn{  n_{R}^{r} C(r) (n_{R}+1)^{-r}= n_{R}^{r}
  D_{U([0,1])}^{\alpha }( \log (n_{R}+1) )  }\qquad \qquad \\ 
& \leq  & e^{rR} D_{U([0,1])}^{\alpha }( R )
\leq (n_{R}+1)^{r} D_{U([0,1])}^{\alpha }( \log n_{R} ) = (n_{R}+1)^{r} C(r) n_{R}^{-r} . \nonumber
\end{eqnarray*}
Letting $R \rightarrow \infty$ yields (\ref{ref_equ_duoe45}) for
$\alpha \in ( 0, r+1 ) \setminus \{ 1 \}$. 
\end{proof}

\begin{proof}[Proof of Theorem \ref{ref_theo_main_resul}]

We divide the proof into four main steps. In step~1 we begin by
proving a (sharp) asymptotic lower bound on the optimal
quantization error for any distribution with  a density that is piecewise
constant on a finite number of intervals of equal lengths. In step~2
we generalize the lower bound of step~1 to any density whose support
is a compact interval on which it is bounded away from zero.
Together with a matching  upper bound based on the companding result
Corollary~\ref{asymp_quant_comp} this will finish the proof for
$\alpha \in (- \infty, 0)$. 
In step 3 we show that the lower bound holds for all distributions
subject to our restrictions and apply again  the companding upper
bound to finish the proof for 
$\alpha \in (0,1)$.
Step 4  treats the remaining  $\alpha = -\infty$ case  and thus
completes the proof. 

Throughout we assume
w.l.o.g.\ that $R \geq R_{0}(C/2)$ where is $C$ defined in
(\ref{const_quant})  and $R_{0}(C/2)$ is  from
Lemma~\ref{lemm_inf_g}.

\noindent \emph{Step 1.} 

Let $M$ be a compact interval of positive length and let $m \geq 2$
and $\alpha \in (-\infty, 1) \setminus \{ 0\}$. Assume that $\mu =
\sum_{i=1}^{m}s_{i} U(A_{i})$, where the $A_{i}$ are disjoint
intervals of equal length $l=l(A_{i})=\lambda(M)/m$ that form a
partition of $M$. We assume $s_i>0$ for all $i=1,\ldots,m$. Thus
$\sum_{i=1}^{m}s_{i} = 1$ and
\[
f= \frac{d \mu}{d \lambda}  = \sum_{i=1}^{m} s_{i} l^{-1} 1_{A_{i}} .
\]
For $\alpha \in (-\infty , 1 ) \setminus \{ 0 \}$  define 
$t_{i} = s_{i}^{1/a_2}  
\left( \sum_{j=1}^{m} s_{j}^{ 
a_1} \right)^{-\frac{1}{1-\alpha }}$, $i=1,\ldots,m$ as in
(\ref{ref_def_ti}).  Let 
\[
R \geq \max \{ 0, \max \{ - \log ( t_{i} ): i = 1,\ldots,m \} \}
\]
and define 
\[
R_{i} = R + \log ( t_{i} ) \geq 0.
\] 
From Proposition~\ref{ref_prop_unit_cube} we deduce
\begin{eqnarray*}
e^{r R} D_{U([0,1])}^{\alpha } (R_{i}) &=&
\left( e^{R-R_{i}} \right)^{r} e^{r R_{i}}
D_{U([0,1])}^{\alpha } (R_{i}) \\
& \rightarrow & t_{i}^{-r} C(r) \qquad \text{as } R \rightarrow \infty .
\end{eqnarray*}
A simple calculation shows
\begin{eqnarray}
\left( \int_{M} f^{a_1}\,  d \lambda \right) ^{a_2} &=&
\left( \int \left( \sum_{i=1}^{m} s_{i}l^{-1}1_{A_{i}}
\right)^{a_1}\,  d \lambda 
\right) ^{a_2} \nonumber \\
&=& l^{(1- a_1)a_2} \left( \sum_{i=1}^{m} s_{i}^{a_1} \right)^{a_2} \nonumber \\
&=& l^{r} \left( \sum_{i=1}^{m} s_{i}^{a_1} \right)^{a_2}
= l^{r} \sum_{i=1}^{m} s_{i} t_{i}^{-r} .
\label{ref_inequ_high_sup_01}
\end{eqnarray}
Now, 
according 
to Lemma~\ref{ref_lemm_gray} there exist functions
$\varepsilon : (0, \infty ) \rightarrow (0, \infty )$ 
and $\delta : (0, \infty ) \rightarrow  ( 0, \infty  )$  such that
  for every $R>0 $ and 
quantizer $q \in \mathcal {Q}$ 
with $H_{\mu}^{\alpha }(q) \leq R$ and $|D_{\mu}^{\alpha }(R) - D_{\mu}(q)| \leq \delta (R)$
we have
\begin{equation}
\label{ref_equ_fsewq}
\max \{ \mu ( q^{-1}(a)  ) : a \in q ( \mathbb{R} ) \} < \varepsilon (R) 
\end{equation}
where $\varepsilon (R) \rightarrow 0 \text{ as } R \rightarrow
\infty$. Moreover,
\begin{equation}
\label{ref_equ_fsewq2}
\max \{  \diam ( A_{i} \cap q^{-1}(a) ) : a \in  q(\mathbb{R}), i \in \{1,\ldots,m \} \} < 
\frac{l \cdot \varepsilon (R)}{\min \{ s_{i} : i=1,\ldots,m \}} .
\end{equation}
Now, again, let $R \geq R_{0}$ and $\gamma > 0$. According to 
Lemma~\ref{lemm_mono}
let $q_{R} \in \mathcal{Q}$ be a quantizer whose codecells with positive $\mu-$mass are intervals,
satisfying $H_{\mu}^{\alpha }(q_{R}) \leq R$ and 
\begin{equation}
\label{eq:deltabound}  
|D_{\mu}^{\alpha }(R) - D_{\mu}(q_{R})| \leq \min ( \gamma e^{-rR},
\delta (R)) .
\end{equation}
Hence, $q_{R}$ satisfies also the relations (\ref{ref_equ_fsewq}) and (\ref{ref_equ_fsewq2}).
In view of Lemma~\ref{lemm_inf_g} let us assume w.l.o.g.\ that $q_{R} \in \mathcal{G}_{R}$ if $\alpha < 0$.
Now let $i \in \{1,\ldots,m\}$
and 
\[
I_{i}(q_{R}) = \{ a \in q_{R}( \mathbb{R} ) :
\lambda\bigl( q_{R}^{-1}(a) \setminus A_i\bigr) =0\}
\]
and 
\[
A_{i,q_{R}} =  \bigcup_{a \in I_{i}(q_{R})} q_{R}^{-1}(a)  .
\]
With
\[
J_{i}(q_{R}) = \{ a \in q_{R}(\mathbb{R}) \setminus I_{i}(q_{R}) : \mu ( A_{i} \cap q_{R}^{-1}(a) ) > 0  \}
\] 
we obtain from (\ref{ref_equ_fsewq2}) that 
\[
\lim_{R \rightarrow \infty }
\sup \{  \diam ( A_{i} \cap q_{R}^{-1}(a) ) : a \in  J_{i}(q_{R}) \} = 0.
\]
Every point of
$J_{i}(q_{R})$ is a codepoints of a codecell which is straddling the boundary of $A_{i}$
and is not $\mu - \text{a.s.}$ contained in $A_{i}$.
Hence $\{A_{i} \cap q_{R}^{-1}(a) : a \in  J_{i}(q_{R})\}$ consists of
at most two intervals and we get 
\begin{equation}
\label{ref_equ_diamcifr}
\lim_{R \rightarrow \infty } \diam ( A_{i,q_{R}} ) = \diam ( A_{i} ) = l.
\end{equation}
We compute
\begin{eqnarray}
D_{\mu }(q_{R}) &=& \sum_{i=1}^{m} s_{i} l^{-1} \int_{A_{i}} | x -
q_{R}(x) | ^{r} \, d \lambda (x) \nonumber  \\
& \geq &  \sum_{i=1}^{m} s_{i} l^{-1} \int_{A_{i,q_{R}}} | x -
q_{R}(x) | ^{r}\,  d \lambda (x) . \label{eq:ailower}
\end{eqnarray}
Let
\begin{eqnarray}
R_{i,q_{R}} &= & H_{U(A_{i,q_{R}})}^{\alpha }(q_{R}) \nonumber
 \\
&=&
\frac{1}{1- \alpha } \log \left( \sum_{a \in I_{i}(q_{R}) }  
( U(A_{i, q_{R}}) ( q_{R}^{-1}(a) ) )^{\alpha } \right) \nonumber  \\*
&=& \frac{\alpha }{\alpha - 1} \left( \log (s_{i}) -
\log \left( \frac{l}{ \lambda (A_{i,q_{R}}) }  \right) \right)
+ \frac{1}{1- \alpha } \log \left( \sum_{a \in I_{i}(q_{R}) } 
\mu ( q_{R}^{-1}(a) ) )^{\alpha } \right)  \label{ref_rift_567}
\end{eqnarray}
where $ U(A_{i, q_{R}}) ( q_{R}^{-1}(a) )$ is the measure of the cell
$ q_{R}^{-1}(a)$ under the uniform distribution on $A_{i, q_{R}}$.
Then using Lemma~\ref{ref_lemm_scale} and (\ref{eq:ailower}) we obtain
\begin{eqnarray}
D_{\mu }(q_{R}) & \geq &   
\sum_{i=1}^{m} s_{i} l^{-1} D_{U(A_{i,q_{R}})}^{\alpha } ( R_{i,q_{R}} ) \diam ( A_{i,q_{R}} ) \nonumber \\
&=& \sum_{i=1}^{m} s_{i} l^{-1} D_{U([0,1])}^{\alpha } ( R_{i,q_{R}} ) \diam ( A_{i,q_{R}} )^{1+r}.
\label{ref_dmufr_435}
\end{eqnarray}
Now pick a sequence $(L_{n})$ of non-negative real numbers, 
such that $L_{n} \rightarrow \infty $,
\begin{equation}
\label{eq:lnlim}
e^{rL_{n}} D_{\mu }^{\alpha }(L_{n}) \rightarrow 
\liminf_{R \rightarrow \infty } e^{rR} D_{\mu }^{\alpha }(R), 
\end{equation}
and
\begin{equation}
\label{eq:lnlimvi}
\frac{ e^{R_{i,q_{L_{n}}}} }{ e^{L_{n}} }  \rightarrow v_{i} \in [0,\infty ], \qquad i=1,\ldots,m
\end{equation}
as $n\to \infty$.  Because we want to determine a lower bound for the
optimal quantization error, using Proposition~\ref{ref_prop_unit_cube}~(i) we can assume w.l.o.g.\ that $q_{L_{n}}$
is $R_{i,q_{L_{n}}}-$optimal for $U(A_{i,q_{L_{n}}})$, i.e., that
$D_{U(A_{i,q_{L_{n}}})}^{\alpha }(R_{i,q_{L_{n}}}) =
D_{U(A_{i,q_{L_{n}}})}(q_{L_{n}})$.  By Proposition~\ref{ref_prop_unit_cube}~(ii) the quantizer $q_{L_{n}}$ divides
$A_{i,q_{L_{n}}}$ into $(k+1)$-intervals with $R_{i,q_{L_n}} \in (
\log (k), \log (k+1)]$ where at least $k$ intervals are of equal
length.  

We next prove that $R_{i,q_{L_{n}}}\to \infty$ as $n\to
\infty$. Assume to the contrary that $(R_{i,q_{L_{n}}})_{n \in \mathbb{N}}$ is
bounded. Then $k=k(R_{i,q_{L_{n}}})$ will also be bounded.  Thus let
$k_{0} \in \mathbb{N}$ such that $k \in \{ 1,\ldots,k_{0} \}$ for
every $n \in \mathbb{N}$.  Together with (\ref{ref_equ_diamcifr}) we
deduce
\begin{eqnarray*}
\liminf_{n \rightarrow \infty }D_{\mu}(q_{L_{n}}) & \geq & 
\liminf_{n \rightarrow \infty } \int_{A_{i,q_{L_{n}}}} | x -
q_{L_{n}}(x) |^{r}  \, d \mu (x) \\
&\geq & \liminf_{n \rightarrow \infty } 
\frac{s_{i}}{l} k \cdot 2 \int_{0}^{\frac{1}{2} \frac{\diam (A_{i,q_{L_{n}}})}{k+1}} x^{r}\, d \lambda (x) \\
&\geq & C(r) s_{i} l^{r} \min \biggl\{ \frac{k}{(k+1)^{r+1}} : k \in \{ 1, \ldots, k_{0} \} \biggr\} > 0.
\end{eqnarray*}
But this contradicts (cf.\  Corollary \ref{ref_coro_lim0}) 
\[
\limsup_{n \rightarrow \infty } 
D_{\mu}(q_{L_{n}}) \leq \limsup_{n \rightarrow \infty } ( D_{\mu }^{\alpha }(L_{n}) + \gamma e^{-r L_{n}}) = 0.
\]
Thus we obtain  that $R_{i,q_{L_{n}}} \rightarrow \infty$  as $n\to
\infty$ for all  $i \in \{1,\ldots,m\}$. 
Proposition~\ref{ref_prop_unit_cube} yields
\begin{equation}
\label{eq:lnoptlimit}
\lim_{n \rightarrow \infty } e^{rR_{i,q_{L_{n}}}} D_{U([0,1])}^{\alpha }
(R_{i,q_{L_{n}}}) = C(r), \qquad i=1,\ldots,m.
\end{equation}
Because $\gamma > 0$ was arbitrary 
we obtain
\begin{eqnarray}
\lefteqn{\liminf_{R \rightarrow \infty } e^{rR} D_{\mu }^{\alpha }(R)
  = \lim_{n \rightarrow \infty } e^{rL_n} D_{\mu }^{\alpha }(L_n)
= \lim_{n \rightarrow \infty } e^{rL_n} D_{\mu }(q_{L_n})
} \qquad  \nonumber \\
&\geq & \lim_{n \rightarrow \infty }
e^{L_{n} r} \sum_{i=1}^{m} s_{i} l^{-1} D_{U([0,1])}^{\alpha } ( R_{i,q_{L_{n}}} ) 
\diam ( A_{i,q_{L_{n}}} )^{1+r} \nonumber  \\
&=& C(r) \sum_{i=1}^{m} s_{i}v_{i}^{-r}  l^{r}
\label{ref_eqn_v001}
\end{eqnarray}
where the first equality holds by (\ref{eq:lnlim}), the second by
(\ref{eq:deltabound}), the inequality follows from
(\ref{ref_dmufr_435}), and the third equality follows from
(\ref{ref_equ_diamcifr}), (\ref{eq:lnlimvi}), and
(\ref{eq:lnoptlimit}).  In the last expression, $1/v_i = 0$ if
$v_i=\infty$. The case $v_{i}=0$ cannot occur because otherwise the
right hand side of (\ref{ref_eqn_v001}) is not finite, which would
contradict the assertion of Proposition~\ref{ref_lemm_pierce}.  Recall
that $\{A_{i} \cap q_{R}^{-1}(a) : a \in J_{i}(q_{R})\}$ contains at
most two intervals for every $i$ and $n \in \mathbb{N}$.  Now assume
that $\alpha \in (0,1)$.  In this case, since by (\ref{eq:interval1})
we can assume w.l.o.g.\ that $H_{\mu}^{\alpha}(q_{L_n})=L_n$, we
obtain
\begin{eqnarray}
1 & \leq & \delta_{1}(L_{n},\mu , q_{L_{n}} ) := \frac{e^{L_{n}(1- \alpha )}}
{  \sum_{i=1}^{m} \sum_{a \in I_{i}(q_{L_{n}})} \mu ( q_{L_{n}}^{-1}(a) )^{\alpha }  } \nonumber \\
& = & \frac
{ \sum_{i=1}^{m} \sum_{a \in I_{i}(q_{L_{n}})} \mu ( q_{L_{n}}^{-1}(a) )^{\alpha } +  
\sum_{a \in q_{L_{n}}(\mathbb{R}) \setminus \cup_{j=1}^{n} I_{j}(q_{L_{n}})} \mu ( q_{L_{n}}^{-1}(a) )^{\alpha }
 }
{  \sum_{i=1}^{m} \sum_{a \in I_{i}(q_{L_{n}})} \mu ( q_{L_{n}}^{-1}(a) )^{\alpha }  } \nonumber \\
& \leq & \frac
{  \sum_{i=1}^{m}       \sum_{a \in I_{i}(q_{L_{n}})} \mu ( q_{L_{n}}^{-1}(a) )^{\alpha } 
+ (m+1) \sup_{a \in q_{L_{n}}(\mathbb{R})} \mu ( q_{L_{n}}^{-1}(a) )^{\alpha }  
  }
{ \sum_{i=1}^{m} \sum_{a \in I_{i}(q_{L_{n}})} \mu ( q_{L_{n}}^{-1}(a) )^{\alpha }  } \nonumber \\
& = &  1 + 
\frac{(m+1) \sup_{a \in q_{L_{n}}(\mathbb{R})} \mu ( q_{L_{n}}^{-1}(a) )^{\alpha }}
{\sum_{i=1}^{m} \sum_{a \in I_{i}(q_{L_{n}})} \mu ( q_{L_{n}}^{-1}(a) )^{\alpha }} .
\nonumber
\end{eqnarray}
From (\ref{ref_rift_567}) and $\lim_{n \rightarrow \infty } R_{i,q_{L_{n}}} = \infty$ we deduce 
\[
\lim_{n\to \infty} \sum_{i=1}^{m} \sum_{a \in I_{i}(q_{L_{n}})} \mu ( q_{L_{n}}^{-1}(a) )^{\alpha } = \infty.
\]
Thus we get 
\begin{equation}
\label{ref_equ_deltedas}
\lim_{n \rightarrow \infty}\delta_{1}(L_{n},\mu , q_{L_{n}} ) = 1.
\end{equation}
Using Lemma~\ref{lemma_straddle_conv} in \ref{appD} we recognize
that the limit relation (\ref{ref_equ_deltedas}) also holds 
for  $\alpha < 0$.
Consequently, we deduce together with (\ref{ref_equ_diamcifr}) and
(\ref{ref_rift_567}) 
for every $\alpha \in (-\infty ,1)\setminus \{0\}$ that
\begin{eqnarray*}
\sum_{i=1}^{m} s_{i}^{\alpha } v_{i}^{ 1 - \alpha } &=&
\sum_{i=1}^{m} s_{i}^{\alpha } \lim_{n \rightarrow \infty } e^{(R_{i,q_{L_{n}}} - L_{n})(1- \alpha )} \\
&=& \lim_{n \rightarrow \infty }
\left( \sum_{i=1}^{m}  
\sum_{a \in I_{i}(q_{L_{n}})} \mu ( q_{L_{n}}^{-1}(a) )^{\alpha } 
\left( \frac{l}{\lambda (A_{i,q_{L_{n}}}) } \right)^{\alpha }
\right) e^{-L_{n} ( 1- \alpha )}  
= 1.
\end{eqnarray*}
Moreover we obtain from (\ref{ref_equ_deltedas}) and
(\ref{ref_rift_567}) that $v_{i} < \infty$ for every $i=1,\ldots,m$. 
Since $\sum_{i=1}^{m} s_{i}^{\alpha } v_{i}^{ 1 - \alpha }=1$, we can
apply  Lemma~\ref{ref_lem_lower_bound},  (\ref{ref_eqn_v001}),  and
(\ref{ref_inequ_high_sup_01}) to obtain for  $\alpha \in (-\infty
,1)\setminus \{0\}$ that  
\[
\liminf_{R \rightarrow \infty } e^{rR } D_{\mu }^{\alpha }(R) \geq 
C(r) \sum_{i=1}^{m} s_{i}t_{i}^{-r}  l^{r}
= C(r) \left( \int_{M} f^{a_1}\, d \lambda \right) ^{a_2} .
\] 

\noindent \emph{Step 2.}

Now let us assume that the support of $\mu$ is  a 
compact interval $M \subset \mathbb{R}$ and
that $f$ is continuous on $M$. 
Again, let $\alpha \in (-\infty, 1) \setminus \{ 0 \}$.
Let $i(f)= \min \{ f(x) : x \in M \}$ resp.\ $s(f)= \max \{ f(x) : x \in M \}$.
Clearly, $s(f)< \infty$. 
Let us assume that
\begin{equation}
\label{ref_essinf0958635}
i(f) > 0.
\end{equation} 
Let  $l=\lambda(M)$. 
For $k \in \mathbb{N}$ partition $M$ into intervals  $\{ A_{i} :
i=1,\ldots,k \}$ 
of common length $\lambda(A_i)=l/k$ 
Set  
\[
\mu _{k} = \sum_{i=1}^{k} \mu ( A_{i} ) U(A_{i})
\]
and
\[
f_{k} = \frac{d\mu_{k}}{d\lambda } = \sum_{i=1}^{k} \frac{\mu (A_{i}) }{\lambda (A_{i})} 1_{A_{i}}.
\]
The continuity of $f$ implies that $f_{k}$ converges pointwise to $f$ as
$k\to \infty$.
In view of (\ref{ref_essinf0958635}) and due to
\begin{eqnarray}
i(f) &=& \min \{ f(x) : x \in M \} \leq  \min \{ f_{k}(x) : x \in M \} \nonumber \\  
& \leq & \max \{ f_{k}(x) : x \in M \} \leq  \max \{ f(x) : x \in M \} = s(f) < \infty
\nonumber 
\end{eqnarray}
for every $k \in \mathbb{N}$, 
dominated convergence implies 
\begin{equation}
\label{ref_equ_hhlambda1}
\lim_{k \rightarrow \infty } \int_{M} f_{k} ^{a_1} \, d\lambda =
\int_{M} f^{a_1}\,  d\lambda .
\end{equation}
Moreover step 1 yields 
\begin{equation}
\label{ref_equ_qralp}
\liminf_{R \rightarrow \infty } e^{rR} D_{\mu_{k} }^{\alpha }(R) 
\geq  C(r) \left( \int_{M} f_{k}^{a_1 }\, d \lambda \right)^{a_2} .
\end{equation}
Now let $R \geq \max(1,R_0)$. Let $\delta > 0$ 
and $q_{R}$ be a quantizer with $|D_{\mu }^{\alpha }(R ) - D_{\mu }(q_{R} )| < \delta e^{-rR}$
and $H_{\mu }^{\alpha }(q_{R}) \leq R$. 
In addition, we assume w.l.o.g.\ (cf.\  Lemma~\ref{pierce_neg_para},
resp.\ Lemma~\ref{lemm_mono})  
that $q_{R} \in \mathcal{K}_{R}(C)$ if $\alpha < 0$ and  $R$ is large
enough,  and
 $H_{\mu }^{\alpha }(q_{R}) = R$ if $\alpha > 0$.
For $i=1,\ldots,k$ let
\[
0 < c_{i,k} = \min\{ f(x) : x \in A_{i} \}  \leq t_{i,k}=  \max
\{ f(x) : x \in A_{i} \}  < \infty  
\]
and
\[
0 < c_{k} = \min \left\{ \frac{c_{i,k}}{t_{i,k}} : i = 1,\ldots,k\right\} .
\]
For every $a \in q_{R}(\mathbb{R})$ we have 
\[
c_{k} \mu_{k} (q_{R}^{-1}(a)) \leq \mu ( q_{R}^{-1}(a) ) \leq 
c_{k}^{-1} \mu_{k} (q_{R}^{-1}(a)) .
\]
and because $f$ is uniformly continuous,
in view of (\ref{ref_essinf0958635}), we have 
\begin{equation}
\label{ref_equ_fracsdf44}
\lim_{k \rightarrow \infty} c_{k} = 1.
\end{equation}
We obtain from the definitions of $H_{\mu }^{\alpha }(q_{R})$ and
$c_k$ that 
\begin{equation}
\label{ref_eqnarr_gdsfg44}
\min \biggl( c_{k}^{\frac{\alpha }{1 - \alpha }} , c_{k}^{\frac{\alpha }{ \alpha - 1}} \biggr) \leq
e^{H_{\mu }^{\alpha }(q_{R}) - H_{\mu_{k}}^{\alpha }(q_{R}) } 
\leq \max \biggl( c_{k}^{\frac{\alpha }{1 - \alpha }} , c_{k}^{\frac{\alpha }{ \alpha - 1}} \biggr) =: v_{k}
\end{equation}
where $v_{k}\to 1$ as $k\to \infty$. Again from the uniform continuity of $f$ we deduce
\begin{equation}
\label{ref_equ_shgnz94sf44}
\lim_{k \rightarrow \infty } \| f - f_{k} \|_{\infty } = 0
\end{equation}
with
\[
\| f - f_{k} \|_{\infty } = \max \{ | f(x) - f_{k}(x) | : x \in M \} .
\]
In view of Proposition~\ref{ref_lemm_pierce} there exists an  $m_0>0$,
such that for all $R\ge1$
\begin{equation}
\label{upp_bou_pierce_m}
D_{\mu }^{\alpha }(R) \leq m_0 e^{-rR}.
\end{equation}
By the choice of $q_{R}$ we have
\begin{equation}
\label{low_b_pr1}
e^{r R} D_{\mu }^{\alpha }(R ) \geq e^{r R }  D_{\mu }(q_{R} ) - \delta .
\end{equation}
Thus (\ref{upp_bou_pierce_m}) yields
\begin{eqnarray*}
| D_{\mu }(q_{R}) - D_{\mu_{k} }(q_{R} )  | 
& \leq  &
\int_{M} | x - q_{R}(x) |^{r} | f(x) - f_{k}(x)  |\, d \lambda (x) \\
& \leq &
\| f - f_{k} \|_{\infty } \frac{1}{i(f)}  \int_{M} | x - q_{R}(x)
|^{r} f(x)\,   d \lambda (x) \\
&\leq &
\| f - f_{k} \|_{\infty } \frac{1}{i(f)} ( D_{\mu }^{\alpha }(R)  + \delta e^{-r R} ) \\
& \leq &
\| f - f_{k} \|_{\infty } \frac{m_0  + \delta}{i(f)} e^{-r R}. 
\end{eqnarray*}
Hence, (\ref{low_b_pr1}) gives 
\begin{eqnarray}
\lefteqn{ e^{r R } D_{\mu }^{\alpha }(R )} \quad  \nonumber \\
& \geq &
e^{r R } ( D_{\mu_k }(q_{R} ) - |D_{\mu }(q_{R} ) - D_{\mu_k }(q_{R} ) | ) - \delta \nonumber \\
& \geq &  
e^{r R }  D_{\mu_{k} }(q_{R} ) - \| f - f_{k} \|_{\infty } \frac{m_0  + \delta}{i(f)}    - \delta  
\nonumber \\
& = &
e^{r ( R - H_{\mu_k}^{\alpha }(q_R) ) } e^{r H_{\mu_k}^{\alpha }(q_R) } D_{\mu_{k} }(q_{R} ) - 
\| f - f_{k} \|_{\infty } \frac{m_0 + \delta}{i(f)} - \delta  . \nonumber \\
& \geq &
e^{- r ( | R - H_{\mu}^{\alpha }(q_R) | + | H_{\mu}^{\alpha }(q_R) - H_{\mu_k}^{\alpha }(q_R) | ) } 
e^{r H_{\mu_k}^{\alpha }(q_R) } D_{\mu_{k} }(H_{\mu_k}^{\alpha }(q_R)
)\nonumber \\ 
& &\mbox{}  - 
\| f - f_{k} \|_{\infty } \frac{m_0 + \delta}{i(f)} - \delta  . 
\label{low_bound_r1}
\end{eqnarray}
Due to the choice of $q_{R}$  
there exists a function $g: (0,\infty) \mapsto (0, \infty)$  with
$e^{ R - H_{\mu }^{\alpha }( q_{R}) } \leq g(R)$ and $g(R) \to 1$ as $R \to \infty$.
Equation (\ref{ref_eqnarr_gdsfg44}) implies 
\begin{equation}
\label{upp_bou_entr_diff_muk}
| H_{\mu_k }^{\alpha }(q_R) - R | \leq 
| H_{\mu_k }^{\alpha }(q_R) - H_{\mu }^{\alpha }(q_R) | 
+ | H_{\mu }^{\alpha }(q_R) - R | 
 \leq 
| \log ( v_{k} ) | + | \log ( g(R) ) | .
\end{equation}
Clearly, inequality (\ref{upp_bou_entr_diff_muk}) yields 
$\lim_{R \rightarrow \infty} H_{\mu_k }^{\alpha }(q_R)= \infty$.
Applying relations (\ref{upp_bou_entr_diff_muk}) 
and (\ref{ref_equ_qralp})
to (\ref{low_bound_r1}) we deduce
\[
\liminf_{R \to \infty } e^{r R} D_{\mu }^{\alpha }(R ) \geq
e^{-r| \log ( v_{k} ) |} 
C(r) \left( \int_{M} f_{k}^{a_1 }\,  d \lambda \right)^{a_2}  -
\| f - f_{k} \|_{\infty } \frac{m_0  + \delta}{i(f)}   - \delta.
\]
By letting $k\rightarrow \infty$ and noting that  $\delta >0$ is
arbitrary we obtain  
from (\ref{ref_equ_hhlambda1}),
(\ref{ref_equ_fracsdf44}) and (\ref{ref_equ_shgnz94sf44})
that 
\begin{equation}
\label{eq_lowerstep2}  
\liminf_{R \rightarrow \infty } e^{r R } D_{\mu }^{\alpha }(R) \geq 
C(r) \left( \int_M f^{a_1 }\,  d \lambda \right)^{a_2}.    
\end{equation} 

Next we show a matching upper bound for $\alpha\in (-\infty,0)$.  The
assumptions on $f$ allow us to use Corollary~\ref{asymp_quant_comp}
showing the existence of a sequence of companding quantizers
$(q_N)=\bigl(Q_{f^*,N} \bigr)$ such that
\begin{equation}
\label{seq_limit}
\lim_{N\to \infty} e^{r H_{\mu}^{\alpha}(q_{N})} D_{\mu}(q_{N}) \leq 
C(r) \left(\int_M f^{a_1}\,  d \lambda \right)^{a_2}.
\end{equation}
Let $R_N= H_{\mu}^{\alpha}(q_{N})$ and note that Proposition \ref{ref_rel_renyi_entr}
implies 
\begin{equation}
\label{RNconv}
 \lim_{N\to \infty} R_N=\infty, \quad   \lim_{N\to \infty}
(R_N-R_{N-1}) =0.  
\end{equation}
Let $R>0$ be arbitrary and let  $n= \max\{N:  R_N\leq R\}$. Then
$R_n\leq R< R_{n+1}$ and since $D_{\mu}^{\alpha}(R)$ is
a nonincreasing function of $R$
\[
e^{rR} D_{\mu}^{\alpha}(R) \le e^{r R_{n+1}} D_{\mu}^{\alpha}(R_n) \leq
e^{r(R_{n+1}-R_n)}  e^{rR_n} D_{\mu}(q_{N}).
\]
This, (\ref{seq_limit}),   and (\ref{RNconv})  yield
\[
\limsup_{R\to \infty} e^{rR} D_{\mu}^{\alpha}(R) \leq  C(r) \left(\int_M f^{a_1} 
\, d \lambda\right)^{a_2} .
\]
Together with (\ref{eq_lowerstep2} ) this completes the proof for the case
$\alpha\in (-\infty,0)$.

\noindent \emph{Step 3.} 

\nopagebreak
Now let $\mu$ be arbitrary, but satisfying all assumptions of the theorem.
Let $\alpha \in (0, 1)$.
For $k,l \in \mathbb{N}$ let 
\[
I_{1} =  ( - \infty , -k ) ,\quad  I_{2} =  [-k,k] \cap
f^{-1}([1/l,l]), \quad
I_{3} = ( k , \infty )
\] 
and
\[
I_{4} = {} \mathbb{R} \setminus ( I_{1} \cup I_{2} \cup I_{3} ).
\]
Because $f$ is bounded and weakly unimodal we can pick $k_{0} \in \mathbb{N}$ such that $\mu ( I_{2} ) > 0$,
$f^{-1}([1/l,l]) = f^{-1}([1/l,\infty ))$, $1/l<l_0$ (see
  Definition~\ref{def_weak_unimod}), and
\[
\mu (I_{2}) = \max \{ \mu (I_{i} ) : i \in \{ 1,2,3,4 \} \}
\]
for every $k \geq k_{0}$ and $l \geq k_{0}$. Note that $I_2$ is a
compact interval.
Now let $\min (k,l) \geq k_{0}$. Let us first assume that $\mu (I_{i} ) > 0$ for
every $i=1,2,3,4$. Consider the  decomposition $\mu = \sum_{i=1}^{4}
\mu (I_{i} ) \mu ( \cdot | I_{i} )$. 
Lemma~\ref{ref_lemm_smu12}  
yields  
\begin{equation}
\label{low_bou_i2}
\liminf_{R \rightarrow \infty } e^{rR} D_{\mu }^{\alpha }(R) \geq 
\mu ( I_{2} ) ^{\frac{1-\alpha + \alpha r}{1- \alpha }} \liminf_{R \rightarrow \infty}
e^{rR} D_{\mu ( \cdot | I_{2} )}^{\alpha } ( R  ).
\end{equation}
By construction,  $i(\mu (I_{2})^{-1}
f1_{I_{2}})>0$, $s(\mu (I_{2})^{-1} f1_{I_{2}})<\infty$, and $\mu
(I_{2})^{-1} f1_{I_{2}}$ is supported by a compact interval. Thus we
can apply the results of step 2. Together with the definition of $a_1$
and $a_2$ we deduce from (\ref{eq_lowerstep2}) and (\ref{low_bou_i2})
that
\begin{eqnarray}
\label{ref_equ_limidjri305}
\liminf_{R \rightarrow \infty } e^{rR} D_{\mu }^{\alpha }(R)  & \geq & 
\mu ( I_{2} )^{\frac{1-\alpha + \alpha r}{1- \alpha }} C(r) 
\left( \int_{I_{2}}( \mu
\bigl(I_{2})^{-1} f\bigr)^{a_1} \, d \lambda  \right)^{a_2} \nonumber
\\ 
& \geq &  
\mu ( I_{2} )^{\frac{1-\alpha + \alpha r}{1- \alpha }-a_1 a_2} C(r) 
\left( \int_{I_{2}} f^{a_1} \, d \lambda  \right)^{a_2} \nonumber \\
&=& C(r) \left( \int_{I_{2}} f^{a_1}\,  d \lambda  \right)^{a_2}.
\end{eqnarray}
Due to $a_1>0 $ and by monotone convergence we obtain
\begin{equation}
\label{ref_equkllim}
\lim_{k \rightarrow \infty } \lim_{l \rightarrow \infty } \int_{I_{2}}
f^{a_1}\,  d \lambda  = 
\int f^{a_1} \, d \lambda.
\end{equation}
Thus we get from (\ref{ref_equ_limidjri305}) that 
\begin{equation}
\label{step3lower}
\liminf_{R \rightarrow \infty } 
e^{rR} D_{\mu }^{\alpha }(R) \geq 
C(r) 
\left( \int f^{a_1}\,  d \lambda  \right)^{a_2}.
\end{equation}
The case $\min \{\mu ( I_{1} ), \mu ( I_{3} ), \mu ( I_{4})
\} = 0$ can be treated similarly.

To show the matching upper bound, note that since $ 
 1/l \le f\le l$ on $I_2$, we can directly apply
 Corollary~\ref{asymp_quant_comp} to $\mu_2:= \mu(\cdot|I_2)$ and its
 density  $f_2:=\mu (I_{2})^{-1}
 f1_{I_{2}}$ to show the existence of a sequence of companding
 quantizers $(q_N)=\bigl(Q_{f_2^*,N} \bigr)$ such that
\[
\lim_{N\to \infty} e^{r H_{\mu_2}^{\alpha}(q_{N})} D_{\mu_2}(q_{N}) \leq 
C(r) \left(\int f_2^{a_1}\,  d \lambda \right)^{a_2}.
\]
Thus by the same argument as in the previous step
\begin{eqnarray}
\label{ref_c2_wer55}
\limsup_{R\to \infty} e^{rR} D_{\mu_2}^{\alpha}(R) &\leq & C(r)
\left(\int f_2^{a_1}  
\, d \lambda\right)^{a_2} \nonumber \\
&= & C(r)\mu ( I_{2} )^{-a_1 a_2}  
\left( \int_{I_{2}} f^{a_1} d \lambda  \right)^{a_2}.
\end{eqnarray}
Again from Lemma \ref{ref_lemm_smu12}  
we obtain for $\alpha \in ( 0,1 )$
the upper bound
\begin{eqnarray*}
\limsup_{R \rightarrow \infty } e^{rR} D_{\mu }^{\alpha }(R) & \leq &  
C(r)\mu ( I_{2} )^{1-a_1 a_2}  t_2^{-r}
\left( \int_{I_{2}} f^{a_1} d \lambda  \right)^{a_2} \\
&& \mbox{}  + 
3 \max_{i=1,3,4} \mu ( I_{i} ) t_{i}^{-r}  
\limsup_{R \rightarrow \infty } e^{rR} D_{\mu ( \cdot | I_{i}) }^{\alpha }(R) .
\end{eqnarray*}
Using Proposition \ref{ref_lemm_pierce} we get a $K > 0$ independent of $k,l$ such that
\[
\limsup_{R \rightarrow \infty } e^{rR} D_{\mu ( \cdot | I_{i}) }^{\alpha }(R) \leq K \mu (I_{i})^{-r/(r + \delta )}.
\]
for $i \in \{1,3,4\}$. 
Letting $l,k$ tend to infinity we obtain by the definition of $t_{i}$ that
$\lim_{l,k \rightarrow \infty} t_{2}^{-r}=1$, resp.\
$\lim_{l,k \rightarrow \infty} t_{i}^{-r} = 0$, $i=1,3,4$. 
Using (\ref{ref_equkllim}) and (\ref{ref_c2_wer55}) we get 
\[
\limsup_{R \rightarrow \infty } e^{rR} D_{\mu }^{\alpha }(R) \leq
C(r) \left( \int f^{a_1} d \lambda  \right)^{a_2}
\]
which, together with the lower bound (\ref{step3lower}) completes the
proof for the case $\alpha\in (0,1)$.  

\noindent\emph{Step 4.} 

Let $\alpha = - \infty$ and $\beta \in (-\infty, 0)$.
Fix $a_1=a_1(\beta )$ and $a_2=a_2(\beta )$. 
From Lemma~\ref{lemm_mono} we deduce
\[
\liminf_{R \rightarrow \infty } e^{rR} D_{\mu }^{-\infty}(R) \geq
\lim_{R \rightarrow \infty} e^{rR} D_{\mu }^{\beta }(R) 
= C(r) \biggl( \int_{\supp (\mu )} f^{a_1} \, d
\lambda\biggr)^{a_2} .
\]
Since the integral on the  right hand side converges to  $ \int f^{1-r}\,
d \lambda$ as $\beta\to -\infty$, we obtain
\[
\liminf_{R \rightarrow \infty } e^{rR} D_{\mu }^{-\infty}(R) \geq
 C(r) \int_{\supp (\mu )} f^{1-r}\,  d \lambda.
\]
The proof is finished by noting that 
Corollary~\ref{asymp_quant_comp} and an argument identical to the one
used in step 2 provide a matching upper bound.
\end{proof}

\section{Concluding remarks}
\label{sec_concl}

We have determined the sharp distortion asymptotics for optimal scalar
quantization with R\'enyi entropy constraint for values $\alpha\in
[-\infty,0)\cup (0,1)$ of the order parameter. Our results, together
  with  the classical $\alpha =0$ and $\alpha=1$ cases, and the recent result
  \cite{Kre10b} for $\alpha \in [r+1,\infty]$,  leave only  open  the
  case  $\alpha \in (1,1+r)$ for which non-matching  upper and
  lower bounds are known to date  (cf.\ \cite{Kre10b}). We note that
  the upper bound provided by 
  optimal companding in
  Corollary~\ref{asymp_quant_comp} also holds for  $\alpha \in
  (1,1+r)$. Based on this, we conjecture that our main result is also
  valid for this remaining range of the $\alpha$ parameter.

Apart from the question of high-rate asymptotics, it remains open if
optimal quantizers exist for all $\alpha \in [-\infty, 1+r ]$. The
non-existence of optimal quantizers in case of $\alpha > 1+r $ has
already been shown in \cite{Kre10a}.  Looking at our main result, it is
obvious that the integrals on the right hand sides of
(\ref{ref_equ_erdmuarf1}) and (\ref{sharp_asymp_alph_neg_inf}) are not
finite in general if $\mu$ has unbounded support. It needs further
research to determine the exact high-rate error asymptotics for
certain classes of source distributions with unbounded support and
$\alpha < 0$. Of special interest is the question whether 
companding quantizers with point density $f^{*}$ are still
asymptotically optimal for source densities with unbounded support.
The definition of $f^{*}$ needs the integrability of $f^{1 /
  a_2}$ in order to guarantee a finite number of quantization
points for the (asymptotically optimal) companding quantizer.
Nevertheless, the right hand side of (\ref{ref_equ_erdmuarf1}) is
defined only when $f^{a_1}$ is integrable.  It remains an open
problem if (\ref{ref_equ_erdmuarf1}) still holds for some $\alpha \in
{} (-\infty ,1+r) \backslash \{1\}$ and distributions where
$f^{a_1}$ is integrable but $f^{1 / a_2}$ is not.
Such an example, if it exists, would show that the companding approach
is not always applicable to generate asymptotically optimal
quantizers, but the known asymptotics (\ref{ref_equ_erdmuarf1}) are
still in force.  Another interesting open question is whether  the
non-integrability of $f^{1/ a_2}$ always implies  the
non-existence of optimal quantizers with a finite codebook.

A careful reading of the proofs shows that many arguments can be
straightforwardly generalized to the $d$-dimensional case and $r$th
power distortion based on some norm on $\mathbb{R}^d$. For $\alpha\in
[-\infty,1)$ and under appropriate conditions we conjecture that
\[
\limsup_{R\to \infty} e^{\frac{r}{d}R} D_{\mu}^{\alpha}(R) =  C(r,d)
\left(\int f^{a_1}  
\, d \lambda^d \right)^{a_2} 
\]
where   $ \lambda^d$ is the $d$-dimensional Lebesgue measure,
\[
a_{1} = \frac{1- \alpha + \alpha \frac{r}{d}}{1- \alpha +
  \frac{r}{d}},
\quad 
a_{2} =  \frac{1-\alpha+\frac{r}{d}}{1-\alpha}
\]
and $C(r,d)$ is a positive constant that depends only on $r$, $d$, and the
underlying norm. 

However, some important steps in our proofs are definitely restricted
to the scalar case, e.g., equation (\ref{ref_equ_onedim}) in Lemma
\ref{ref_lemm_gray}, which yields (\ref{ref_equ_diamcifr}).  One of
the key problems concerns the first step of the proof of
Theorem~\ref{ref_theo_main_resul}.  In higher dimensions one has to
control the contribution to distortion and entropy of cells straddling
the common boundary of at least two touching cubes in the support of
$\mu$.  The ``firewall'' construction used in case of $\alpha = 0$
(see \cite[p.87]{GrLu00}) does not seem to work  in the
general case. For $\alpha \neq 0$ it seems to be very hard to control
the entropy of the quantizer when adding or changing codecells and
codepoints in a certain region.  In order to progress in this
direction, one would certainly need more refined knowledge about the codecell
geometry of (asymptotically) optimal quantizers.  Even in the case
$\alpha = 0$ little is known on this subject (results in \cite{Sag08}
highlight the difficulty of the problem).  As already mentioned in the
introduction, the methods used for the case $\alpha = 1$ are also not
applicable to the general case because they rely on the special
functional form of the Shannon entropy.  It appears that
generalization to higher dimensions would necessitate the development
of  isodiametric inequalities for the (bounded) codecells of
asymptotically optimal quantizers.

\appendix
\appsec
\label{appA}

\noindent\emph{Proof of Lemma~\ref{lemm_mono}.} \  To show
(\ref{eq:mono}), let $q \in \mathcal{Q}$ be such that
$H_{\mu}^{\alpha}(q) \leq R$ and assume $\beta \leq \alpha$.  It is
easy to check that $\frac{d}{d\gamma} H_{\mu}^{\gamma}(q)\le 0$ on
$(-\infty,0)\cup (0,1)\cup (1,\infty)$, and thus the mapping $ \gamma
\mapsto H_{\mu}^{\gamma}(q)$ is non-increasing on these intervals. In
view of the continuity  of  $H_{\mu}^{\alpha}(q)$ at $\alpha\in
\{0,1\}$ (see (Remark~\ref{remark_hospital}(a)) we deduce that
$H_{\mu}^{\alpha}(q) \leq H_{\mu}^{\beta}(q)$.  Now the assertion
follows from Definition (\ref{ref_darst_f_q}).

Equation (\ref{eq:interval}) of the second statement follows directly
from the more general results Theorem~3.2 and Proposition~4.2 in
\cite{Kre11}.  For (\ref{eq:interval1}), we refer to \cite[Proposition
  2.1.(i)]{Kre10a}.  \qed

\medskip

\noindent\emph{Proof of Proposition~\ref{bennet_int}.} \ \ We proceed
in several steps.
\smallskip \\ \emph{1.} \ Since $\hat{G}$ is increasing, it has a derivative
$\hat{G}'$ a.e.\ (by convention we set $\hat{G}'(x)=0$ if $\hat{G}$ is
not differentiable at $x$). Also, note that $G$ and $\hat{G}$ are strictly
increasing on $I$, resp.\ on $[0,1]$, and $\hat{G}$ is Lipschitz with
constant $(\mathrm{ess } 
\inf_{I} g )^{-1}$ and thus absolutely continuous. Since
$\hat{G}^{\prime}( x ) = 1/g(\hat{G}(x))$ a.e.\ on $(0,1)$, we obtain
\[
\int \frac{f}{g^{r}} \, d \lambda = \int (\hat{G}^{\prime }(G(x)))^{r}\, d \mu (x).
\]
\smallskip \\
\emph{2.} \  Next we prove
\begin{equation}
\label{conv_mn}
\lim_{N \rightarrow \infty} \int (m_{N}(G(x)))^{r}\,  d \mu (x) =
\int (\hat{G}^{\prime }(G(x)))^{r} \, d \mu (x),
\end{equation}
where $m_{N}$ is the piecewise constant function defined by  $m_{N}(x)= N
\int_{[(i-1)/N,i/N)} \hat{G}^{\prime }\, d \lambda$, if $x\in[(i-1)/N,i/N)$,
    $i=1,\ldots,N$ and $m_N(x)=0$ otherwise.

Lebesgue's differentiation theorem (see, e.g.,
\cite[Thm.\  6.2.3]{Coh80}) implies 
that $m_{N}(x) \rightarrow \hat{G}^{\prime }(x)$ as $N \rightarrow \infty$
 a.e.\ on  $(0,1)$.
Also, from $\hat{G}^{\prime}( x ) = 1/g(\hat{G}(x))$ we deduce 
\begin{eqnarray*}
m_{N}(\cdot ) 
& \leq & 
\max \biggl\{ N \int_{[(i-1)/N,i/N)} \frac{1}{g(\hat{G}(x))}\, d \lambda (x) :
  1\le  i\leq N  \biggr\} \nonumber \\
& \leq &
(\mathrm{ess } \inf\nolimits_{I} g)^{-1}. 
\end{eqnarray*}
Thus, the dominated convergence theorem  yields (\ref{conv_mn}). \smallskip \\
\emph{3.} \  Let $\bar{Q}_{g,N}$ be the quantizer with the
same codecells as $Q_{g,N}$ but with the midpoints of the codecells as
quantization points:
\[
\bar{Q}_{g,N}(x)=  \frac{1}{2}( \hat{G}(i/N) + \hat{G}((i-1)/N) )
  \quad \text{if 
  $x\in I_{i,N}$,\ \ }  i=1,\ldots,N.
\]
We will show that 
\[
\lim_{N \rightarrow \infty } N^{r} \int | x - \bar{Q}_{g,N}(x) |^{r} f_{N}(x)\, d \lambda(x) 
= C(r) \int \frac{f}{g^{r}}\, d \lambda, \quad \quad 
\]
where $f_N$ is the piecewise constant density defined by $f_{N}(x) =
\frac{1}{\lambda(I_{i,N})} \int_{I_{i,N}} f \, d \lambda$ if $x\in
I_{i,N}$, $i=1,\ldots,N$ and  $f_N(x)=0$ otherwise. 

A simple calculation shows
\begin{eqnarray*}
\lefteqn{ N^{r} \int | x - \bar{Q}_{g,N}(x) |^{r}
  f_{N}(x)\, d \lambda(x) }\quad \quad \\ 
&=& N^{r} \sum_{i=1}^{N} C(r) f_{N}\left(\hat{G}\biggl(\frac{i}{N}\biggr)\right) ( \lambda ( I_{i,N} ) )^{r+1} \\
&=& C(r) N^{r} \sum_{i=1}^{N} \int_{I_{i,N}} f(x) (N^{-1}m_{N}(G(x)))^{r}\, d \lambda(x) \\
&=&  C(r) \int m_{N}(G(x))^r f(x)\,  d \lambda(x).
\end{eqnarray*}
Now the assertion follows from steps 1 and  2.
\smallskip \\
\emph{4.}  \  Next we show that 
\[
\lim_{N \rightarrow \infty } N^{r} D_{\mu }(\bar{Q}_{g,N}) = C(r) \int
\frac{f}{g^{r}} \, d \lambda.
\]

For any $i \in \{ 1,\ldots,N \}$ and $x \in I_{i,N}$ we have 
\[
 N^{r} |x - \bar{Q}_{g,N}(x) |^{r}  \leq 
N^{r}\bigl(\lambda(I_{i,N})\bigr)^r = 
N^{r} \left( \int_{((i-1)/N, i/N)} \hat{G}^{\prime }\, d \lambda\right)^{r} =
m_{N}(G(x))^{r}. 
\]
Therefore
\begin{eqnarray*}
\lefteqn{ \biggl| N^{r} \int | x - \bar{Q}_{g,N}(x) |^{r} f(x)\, d \lambda(x) -
  N^{r} \int | x - \bar{Q}_{g,N}(x) |^{r} f_{N}(x)\, d
  \lambda(x) \biggr| } \qquad \qquad  \\
& \leq &
\int
  m_N(G(x))^r |f(x)-f_N(x)|\, d \lambda(x) \nonumber \\
&\le &
(\mathrm{ess } \inf\nolimits_{I} g)^{-1}   \int |f(x)-f_N(x)|\, d \lambda(x).
\end{eqnarray*}
By Lebesgue's differentiation theorem we have $f_{N} \rightarrow f$
a.e., and now  Scheff\'e's theorem  \cite[Thm.\  16.11]{Bil86}
implies 
\[
 \lim_{N\to \infty} \int |f-f_N|\, d\lambda =0.
\]
Hence,
\[
\lim_{N \rightarrow \infty } N^{r} D_{\mu }(\bar{Q}_{g,N}) = 
\lim_{N \rightarrow \infty }
N^{r} \int| x - \bar{Q}_{g,N}(x) |^{r} f_{N}(x)\, d \lambda(x),
\]
where the right hand side is equal to $ C(r) \int \frac{f}{g^{r}}\,
d \lambda$ from step 3.
\smallskip \\
\emph{5.} \  In view of step 4,  to  prove relation (\ref{bennet}) it suffices
to show  that
\begin{equation}
\label{mid_vs_general_qpoint}
\lim_{N \rightarrow \infty } N^{r} D_{\mu }(\bar{Q}_{g,N}) =
\lim_{N \rightarrow \infty } N^{r} D_{\mu }(Q_{g,N}).
\end{equation}

Applying the mean value theorem of differentiation (if $r>1$) or by
the triangle inequality (if $r=1$), we have for each $i \in \{
1,\ldots,N \}$ and $x\in I_{i,N}$,
\begin{eqnarray}
\lefteqn{ N^{r}\bigl| | x - \bar{Q}_{g,N}(x) |^{r} - | x - Q_{g,N}(x)
  |^{r} \bigr|} \nonumber\qquad  \\ 
& \leq &
N^{r} | \bar{Q}_{g,N}(x) - Q_{g,N}(x) | r ( \lambda(I_{i,N}))^{r-1} .
\label{upp_bou_mean_v_th}
\end{eqnarray}
Further, note that the definitions of $m_{N}$, $Q_{g,N}$, $\bar{Q}_{g,N}$
also yield
\begin{eqnarray}
\lefteqn{N^{r}\bigl| | x - \bar{Q}_{g,N}(x) |^{r} - | x - Q_{g,N}(x)
  |^{r} \bigr |} \nonumber \qquad \\ 
& \leq &
r N^{r} ( \lambda(I_{i,N}))^{r} = r \cdot m_{N}(G(x))^{r}. 
\label{upp_bou_dom_conv}
\end{eqnarray}

Let $I_{i,N}^1$ and $I_{i,N}^2$ denote the partition of $I_{i,N}$ into
two intervals of equal length $\lambda(I_{i,N})/2$. Let $j=j(x) \in \{1,2\}$
be such that $x \in I_{i,N}^{j}$. Letting  $a=\inf (I_{i,N})$ and $b =
\sup ( I_{i,N} )$ we obtain by the absolute continuity of $\hat{G}$ 
\[
\hat{G}((a+b)/2) = \hat{G}(a) + \int_{a}^{(a+b)/2} \hat{G}^{\prime }\,  d \lambda
\]
and
\[
\frac{\hat{G}(a) + \hat{G}(b)}{2} = \hat{G}(a) + \int_{a}^{(a+b)/2} \frac{\hat{G}(b) -
  \hat{G}(a)}{b-a} \, d \lambda .
\]
Thus we get 
\begin{eqnarray}
| \bar{Q}_{g,N}(x) - Q_{g,N}(x) | 
&=& 
\biggl| \hat{G}((a+b)/2) - \frac{\hat{G}(a) + \hat{G}(b)}{2} \biggr| \nonumber \\
& \leq &
\lambda(I_{i,N}^{j}) |L(x,N)|
\label{diff_q_point}
\end{eqnarray}
where
\begin{eqnarray*}
L(x,N) &=& \frac{1}{\lambda(I_{i,N}^{j})} \int_{I_{i,N}^{j}} 
\hat{G}^{\prime }  \, d \lambda
- 
\frac{\hat{G}(b)-\hat{G}(a)}{b-a}  \\
&=& 
\frac{1}{\lambda(I_{i,N}^{j})} \int_{I_{i,N}^{j}} 
\hat{G}^{\prime } \, d \lambda 
- 
\frac{1}{\lambda(I_{i,N})} \int_{I_{i,N}} 
\hat{G}^{\prime } \, d \lambda
\end{eqnarray*}
if $x\in I_{i,N}^j$, $i=1,\ldots,N$.  
In view of (\ref{upp_bou_mean_v_th}) and (\ref{diff_q_point}) we deduce
\begin{eqnarray*}
N^{r}\bigl | | x - \bar{Q}_{g,N}(x) |^{r} - | x - Q_{g,N}(x) |^{r} \bigr|
& \leq &
r N^{r} \lambda(I_{i,N}^{j})| L(x,N)| ( \lambda(I_{i,N}))^{r-1}
\nonumber \\ 
& \leq &
r \cdot m_{N}(G(x))^{r}| L(x,N)| . 
\end{eqnarray*}
Lebesgue's differentiation theorem yields $\lim_{N \to\infty} L(x,N) =
0$ a.e. Hence
\begin{equation}
\label{conv_zero}
\lim_{N\to \infty} N^{r}\bigl| | x - \bar{Q}_{g,N}(x) |^{r} - | x
- Q_{g,N}(x) |^{r} \bigr| = 0 \quad \text{a.e.}
\end{equation}
Due to the relations (\ref{conv_zero}), (\ref{upp_bou_dom_conv}) and 
together with step 2, we can apply the generalized dominated convergence theorem
\cite[Chapter 11.4]{Roy68} 
to obtain (\ref{mid_vs_general_qpoint}). \qed

\begin{lemma}
\label{dens_inv_compr}
Let $\mu$ be a probability distribution which is absolutely continuous with
respect to $\lambda$ and let $f$ denote its density.
Let $G$ be a compressor for $\mu$ with point density $g$. 
If $\{g=0\} \subset \{ f = 0 \}$,
then
$\mu \circ G^{-1}$ is absolutely continuous with
respect to $\lambda$.
Also 
\begin{equation}
\label{h_density}
\hat{G}^{\prime }(y )= \frac{1}{g(\hat{G}(y ))} 1_{\{ g>0 \}}(\hat{G}(y))
\quad   \text{a.e. on $(0,1)$} 
\end{equation}
and $\mu \circ G^{-1}$ has the density 
\begin{equation}
\label{density}
f_{G}(y) = f(\hat{G}(y ))\hat{G}'(y) =  \frac{f(\hat{G}(y ))}{g(\hat{G}(y))} 1_{ \{g> 0 \} }(\hat{G}(y ))
\quad   \text{a.e. on $(0,1)$.}
\end{equation}
\end{lemma}
\begin{proof}
In order to prove that $\mu \circ G^{-1}$ is absolutely continuous
let us make 
the key observation that, although $G$ is in general not invertible,   we have
\begin{equation}
\label{invert}
\hat{G}(G( x ))= x \quad \mu \text{-a.e.\ } x \in \mathbb{R}  
\end{equation}
Indeed, by the definition of $G$ and due to $\{g=0\} \subset \{ f = 0 \}$
there exists a measurable set $A_{G} \subset \mathbb{R}$ such that
$G$ is differentiable on $A_{G}$, $\mu (A_{G})=1$,  and
\[
G^{\prime }(x) = g(x) \in (0,\infty ) \text{ for every } x \in A_{G} .
\] 
Hence $G$ is locally invertible at $x \in A_{G}$, so $\hat{G}(G(x))=x$
which proves (\ref{invert}).
Moreover,
\begin{equation}
\label{diff_inv_dens_g}
\hat{G}^{\prime}(G(x)) = 1/ g(x) \text{ for every } x \in A_{G} 
\end{equation}
which proves (\ref{h_density}).

$\hat{G}$ is strictly increasing (and
thus one-to-one) and maps $(0,1)$ onto $\hat{G}((0,1))$.
Thus, together with  
(\ref{invert}) we obtain for every Borel measurable $B \subset \mathbb{R}$ that
\begin{equation}
\label{absc}    
\mu \circ G^{-1}(B) = \mu ( \{ x : \hat{G}(G(x)) \in \hat{G}(B) \} ) = \mu ( \hat{G}(B) ). 
\end{equation}
If $U([0,1])$ denotes the uniform distribution on $[0,1]$ we obtain
again from
(\ref{invert}) that
\[
U([0,1]) \circ \hat{G}^{-1}((- \infty, x]) = G(x) \quad \text{for\ }
 \text{a.e.\ } x \in \mathbb{R}. 
\]
Thus, $U([0,1]) \circ \hat{G}^{-1} = g \lambda$. 
Now let $B \subset {} (0,1)$ be Borel measurable and
$\lambda(B)=0$. This implies $U([0,1])\circ \hat{G}^{-1}(\hat{G}(B)) = 0$.
Because  $\mu$  is absolutely continuous with respect to $ g\lambda$
we obtain $\mu ( \hat{G}(B) ) = 0$. Hence, (\ref{absc}) implies $\mu \circ G^{-1}(B)=0$ showing
that $\mu \circ G^{-1}(B)$ is absolutely continuous with respect to $\lambda$.
In order to prove (\ref{density}) let $[a,b] \subset {} (0,1)$.
In view of (\ref{absc}) and from the definition of $\hat{G}$ we obtain
\begin{equation}
\label{dens_ident}
\mu ( G^{-1}([a,b]) ) = \int_{\hat{G}(a)}^{\hat{G}(b)} f \, d \lambda.
\end{equation}
From (\ref{diff_inv_dens_g}) we deduce 
\[
\hat{G}^{\prime }(y)= \frac{1}{g(\hat{G}(y))} 1_{\{ g >0 \}}(\hat{G}(y)) \quad
\text{a.e.\ on $(0,1)$.} 
\]
Because $\mu \circ G^{-1}$ is absolutely continuous with 
respect to $\lambda$ its cumulative distribution function
is absolutely continuous and, therefore, differentiable a.e.
Applying the chain rule for the Lebesgue integral (see \cite[Corollary 4]{SeVa69}) 
we obtain
\begin{equation}
\label{int_ident}
\int_{\hat{G}(a)}^{\hat{G}(b)} f(x)\,  d \lambda(x) = \int_a^b f(\hat{G}(y))\hat{G}'(y) \, d \lambda(y).
\end{equation}
Now, (\ref{dens_ident}) and (\ref{int_ident}) prove the first equation in (\ref{density}).
The second equality  in (\ref{density}) follows from (\ref{h_density}).
\end{proof}

\noindent\emph{Proof of Lemma~\ref{ref_lemma_xyz}.} 

\noindent \emph{1.} $\alpha \in {} (0,\infty)$.  For this range of
$\alpha$ the result goes back to R\'enyi \cite[11\S]{Ren60a} who
stated it with somewhat less generality. Csisz\'ar
\cite[Thm.\ 2]{Csi73} gives a  more general form of the
result that implies our statement. 

\noindent \emph{2.} $\alpha = - \infty$. 
\nopagebreak

Clearly,
\begin{equation}
\label{low_bou_ess_inf}
\liminf_{\Delta \rightarrow 0} 
\frac{\inf \{ \mu(\hat{q}_{\Delta, M }^{-1}(a)) : a \in \hat{q}_{\Delta, M}(\mathbb{R})  \} }{\Delta }
\geq 
\mathrm{ess } \inf{} _{M} f .
\end{equation}
Now let $\varepsilon > 0$ and define $N_{\varepsilon} = \{x: f(x) <  \mathrm{ess } \inf{} _{M} f  + \varepsilon \} \cap M$.
Hence, $\lambda(N_{\varepsilon}) > 0$. By Lebesgue's differentiation theorem 
we can find an $x \in N_{\varepsilon}$ such that $\mu$ is differentiable at $x$ with $f(x) = \frac{d \mu }{d \lambda}(x)$.
Moreover a $\Delta_{0}(\varepsilon) > 0$ exists, such that for every $\Delta \leq \Delta_{0}$ a
$b \in \hat{q}_{\Delta, M}(\mathbb{R})$ can be found with $x \in \hat{q}_{\Delta, M}^{-1}(b)$ and
\[
\frac{\mu ( \hat{q}_{\Delta, M}^{-1}(b) )}{\Delta } \leq \mathrm{ess } \inf{} _{M} f  + 2 \varepsilon .
\]
Because $\varepsilon$ is arbitrary we obtain
\begin{equation}
\label{upp_bou_ess_inf}
\limsup_{\Delta \rightarrow 0} 
\frac{\inf \{ \mu(\hat{q}_{\Delta, M }^{-1}(a)) : a \in \hat{q}_{\Delta, M}(\mathbb{R})  \} }{\Delta }
\leq 
\mathrm{ess } \inf {}_{M} f .
\end{equation}
In view of Definition~\ref{ref_Def_renyi_entr_diff}
and the definition of $H_{\mu }^{- \infty }(\cdot)$, the
combination of (\ref{low_bou_ess_inf}) and (\ref{upp_bou_ess_inf}) yields the assertion. 

\noindent \emph{3.} $\alpha \in {} ( - \infty, 0 ]$. 
Here we adapt R\'enyi's original proof to our case.  With the convention
$0^{0}:= 0$ and in view of Definition \ref{ref_Def_renyi_entr_diff}
resp.\ Remark \ref{remark_hospital} it suffices to show that
\begin{equation}
\label{alpha_entr_conv}
\int_{\supp (\mu )} f^{\alpha }\, d \lambda =
\lim_{\Delta \rightarrow 0} \sum_{a \in \hat{q}_{\Delta }(\mathbb{R})} \Delta^{1-\alpha } 
\left( \int_{\hat{q}_{\Delta }^{-1}(a)} f\, d \lambda \right)^{\alpha } .
\end{equation}
For $\Delta > 0$ and $x \in \mathbb{R}$ we define
\[
g_{1, \Delta }(x ) =  1_{\supp (\mu)}(x)\sum_{a \in \hat{q}_{\Delta }(\mathbb{R})} 1_{\hat{q}_{\Delta}^{-1}(a)} (x) \frac{1}{\Delta} 
\int_{\hat{q}_{\Delta}^{-1}(a)}  f^{\alpha }\, d \lambda  
\]
and
\begin{eqnarray*}
g_{2, \Delta }(x ) &=&  1_{\supp (\mu)}(x) \sum_{a \in \hat{q}_{\Delta }(\mathbb{R})}
1_{\hat{q}_{\Delta}^{-1}(a)} (x) \Delta^{- \alpha }  
\left( \int_{\hat{q}_{\Delta}^{-1}(a)} f\, d \lambda \right)^{\alpha } \\
&=&  1_{\supp (\mu)}(x) \left( \sum_{a \in \hat{q}_{\Delta }(\mathbb{R})} 1_{\hat{q}_{\Delta}^{-1}(a)}(x)
\frac{1}{\Delta}  
\int_{\hat{q}_{\Delta}^{-1}(a)} f\,  d \lambda \right)^{\alpha }. 
\end{eqnarray*}
Applying Lebesgue's differentiation theorem 
we obtain $g_{1, \Delta } \rightarrow f^{\alpha}$ and $g_{2, \Delta}
\rightarrow f^{\alpha}$ a.e.\  as $\Delta \rightarrow 0$.
Now note that since $\alpha \le 0$, the function $x\mapsto x^{\alpha}$
is convex on $(0,\infty)$, so by Jensen's inequality 
\[
 \left( \frac{1}{\Delta}  
\int_{\hat{q}_{\Delta}^{-1}(a)} f\,  d \lambda \right)^{\alpha } \le 
\frac{1}{\Delta}    \left( \int_{\hat{q}_{\Delta}^{-1}(a)} f^{\alpha}\,  d
\lambda \right)
\]
for all $a\in q^{-1}(\R)$, implying $g_{2,\Delta}(x)\le
g_{1,\Delta}(x)$ for all $x$. Since $g_{2,\Delta}\ge 0$ and   $\int
g_{1, \Delta }\, d \lambda = \int_{\supp (\mu)} f^{\alpha }\, d \lambda
\in {} (0, \infty )$,  we can apply the generalized
dominated convergence theorem \cite[Chapter 11.4]{Roy68} to obtain
\[
\lim_{\Delta \rightarrow 0}  \int g_{1, \Delta }\, d \lambda = 
\lim_{\Delta \rightarrow 0}  \int g_{2, \Delta }\, d \lambda, 
\]
which is equivalent to (\ref{alpha_entr_conv}) 
and, therefore, finishes the proof.\qed 

\appsec
\label{appB}

\noindent\emph{Proof of Proposition~\ref{ref_lemm_pierce}.}

(i) Recall that $\lfloor x\rfloor $ denotes the largest
integer less than or equal to $x\in \R$.  In view of 
Lemma~\ref{lemm_mono} we have  $D_{\mu }^{\alpha }(R) \leq D_{\mu }^{0}(R)$.
Consequently, we deduce from \cite[Lemma 1]{LuPa08} 
the existence of a constant $\kappa > 0$  
(that depends only on $r$ and $\delta$) such that
\[
D_{\mu }^{\alpha }(R) \leq 
D_{\mu }^{0 }(R) \leq
D_{\mu }^{0 }( \log ( \lfloor e^{R} \rfloor ) ) \leq ( \lfloor e^{R} \rfloor )^{-r} 
\kappa^{r} \left( \int | x | ^{r+ \delta } \, d\mu (x) \right)^{r/(r + \delta )}.
\]  
Due to $R \geq 1$ we obtain
\[
( \lfloor e^{R} \rfloor )^{-r} \leq e^{-rR } \left( \frac{e^{R}}{e^{R}-1} \right)^{r}
\leq e^{-rR } \left( \frac{e}{e-1} \right)^{r},
\]
which yields the assertion with $C_0=\kappa^r \left( \frac{e}{e-1}
\right)^{r}$. \smallskip \\ 
(ii)  In view of Lemma~\ref{lemm_mono} it
is enough to prove relation (\ref{upp_bou_neg_alp}) for $\alpha = -
\infty$.  Let $I=\supp(\mu)$,  $R \geq 0$ and $q_{R} \in \mathcal{Q}$ with
$H_{\mu}^{-\infty }(q_{R}) \leq R$.  According to
Lemma~\ref{lemm_mono} let us assume  
w.l.o.g.\ that all codecells of $q_{R}$ with positive $\mu-$mass are
intervals.  By subdivision of codecells with $\mu-$mass greater than
or equal to $2 e^{-R}$ we can assume w.l.o.g.\ that
\begin{equation}
\label{inequ_intsec}
e^{-R} \leq \mu ( q_{R}^{-1}(a) ) < 2 e^{-R}
\end{equation}
for every $a \in q_{R}(\mathbb{R})$ with $\mu(q_{R}^{-1}(a))>0$, where
the first inequality holds since $H_{\mu}^{-\infty}(q_R)\le R$.
Moreover, for every such $a$ we obtain
\begin{equation}
\label{bou_leng_meas}
\diam (  q_{R}^{-1}(a) \cap I )
\leq 
\frac{\mu (q_{R}^{-1}(a))}{i(f)} 
\end{equation}
and we can assume w.l.o.g.\ that $a \in q_{R}^{-1}(a)  \cap I$ if
$q_{R}^{-1}(a) ) \cap I \neq \emptyset$ (otherwise the distortion can
be decreased by redefining $a$).
Then we have
\begin{eqnarray*}
D_{\mu }(q_{R}) &=& \sum_{a \in q_{R}(\mathbb{R})}
\int_{q_{R}^{-1}(a)} |x-a|^{r} f(x) \, d\lambda(x) \\
& \leq & \sum_{a \in q_{R}(\mathbb{R})} \int_{q_{R}^{-1}(a)} 
( \diam ( q_{R}^{-1}(a) \cap I ) )^{r} f(x) \, d\lambda(x).
\end{eqnarray*}
In view of  (\ref{inequ_intsec}) and (\ref{bou_leng_meas}) we get 
\begin{eqnarray*}
D_{\mu }(q_{R}) & \leq & \sum_{a \in q_{R}(\mathbb{R})} 
\int_{q_{R}^{-1}(a)} \left( \frac{\mu (q_{R}^{-1}(a))}{i(f)}
\right)^{r} f(x) \, d\lambda(x) \nonumber \\ 
&=& \frac{1}{i(f)^{r}} \sum_{a \in q_{R}(\mathbb{R})} (\mu (q_{R}^{-1}(a)))^{r+1} \nonumber \\
&<& \frac{1}{i(f)^{r}} \sum_{a \in q_{R}(\mathbb{R})} \mu (q_{R}^{-1}(a)) (2 e^{-R})^{r} = 
\frac{2^{r}}{i(f)^{r}} e^{-rR},
\end{eqnarray*}
which yields (\ref{upp_bou_neg_alp}) by taking the infimum over all $q_R \in \mathcal{Q}$ 
with $H^{-\infty}_\mu(q_R)\le R$. \qed 

\medskip

\noindent\emph{Proof of Lemma~\ref{ref_lemm_gray}.} \  Let
$\varepsilon > 0$. Choose $c, t \in (0,\infty)$, such that
\begin{equation}
\label{iequ_ac_at_2}
1 - \mu ( A_{c,t}) < \frac{\varepsilon }{2} 
\end{equation}
with $A_{c,t}=\{ x : f(x) \geq c \}\cap \{ x : f(x) \leq t \} \cap [-t,t]$.
Let
\[
\kappa = \frac{c}{(1+r)2^{r}}
\]
and use  Corollary \ref{ref_coro_lim0} to  choose $R_{0} > 0$ such that
\[
t \left(  \frac{2 D_{\mu}^{\alpha } (R_{0}) }{\kappa}  \right)^{\frac{1}{1+r}} < \frac{\varepsilon }{2} . 
\]
Now let $R \geq R_{0}$, $\delta = D_{\mu}^{\alpha } (R) > 0$,  and choose $q \in \mathcal{Q}$ with
$H_{\mu}^{\alpha }(q) \leq R$ and
$| D_{\mu }(q) - D_{\mu }^{\alpha }(R) | < \delta$.
We have
\begin{eqnarray}
D_{\mu }(q) &=& 
\sum_{a \in q(\mathbb{R})} \int_{q^{-1}(a)} | x -a | ^{r} f(x)\, d\lambda (x) \nonumber \\
& \geq &  
\sum_{a \in q(\mathbb{R})} c \int_{q^{-1}(a) 
\cap A_{c,t} } | x -a | ^{r}\, d\lambda(x) .
\label{err_ac_at}
\end{eqnarray}
Let $B(x,l)=[x-l,x+l]$ for any $l>0$ and $x\in \R$. 
For every $a \in q ( \mathbb{R} )$ define 
 \begin{equation}
\label{ref_f1acap_44}
s_{a} = \lambda( q^{-1}(a) \cap A_{c,t})/2.
\end{equation}
Since  $A_{c,t}$ is bounded, we have $s_a\in [0,\infty)$. Moreover, it is easy to show that
\begin{equation}
\label{qm1_ac_at}
\int_{q^{-1}(a) \cap A_{c,t}} | x -a | ^{r}\, d\lambda (x) \geq  
\int_{B(a,s_{a})} | x -a | ^{r} \, d\lambda (x)
\end{equation}
(see, e.g.,  \cite[Lemma 2.8]{GrLu00}).
Using (\ref{ref_f1acap_44}) we compute 
\begin{equation}
\label{ref_basa_34453}
\int_{B(a,s_{a})} | x -a | ^{r}\, d\lambda (x) = 
\frac{2 s_{a}^{r+1}}{1+r}   = 
\frac{s_{a}^{r}}{1+r} \lambda( q^{-1}(a) \cap A_{c,t} ) .
\end{equation}
Combining (\ref{qm1_ac_at}) and (\ref{ref_basa_34453})
with (\ref{err_ac_at}) we obtain
\begin{eqnarray*}
D_{\mu }(q) & \geq & 
\frac{c}{1+r} \sum_{a \in q(\mathbb{R})} \lambda( q^{-1}(a) \cap A_{c,t} ) s_{a}^{r} \\
&=&  \frac{c}{(1+r)2^{r}} 
\sum_{a \in q(\mathbb{R})} \lambda( q^{-1}(a) \cap A_{c,t} ) ( 2 s_{a} )^{r} .
\end{eqnarray*}
Using (\ref{ref_f1acap_44}) we get
\begin{eqnarray*}
D_{\mu }(q) & \geq & 
\frac{c}{(1+r)2^{r}} 
\sum_{a \in q(\mathbb{R})} \lambda ( q^{-1}(a) \cap A_{c,t} ) ^{1+r} \\
&\geq & \kappa \cdot \sup_{a \in q(\mathbb{R})}  \lambda ( q^{-1}(a) \cap
A_{c,t} ) ^{1+r}  . 
\end{eqnarray*}
On the other hand the choice of $\delta$ and the monotonicity of $D_{\mu}^{\alpha } (\cdot )$ yields
\[
D_{\mu }(q) \leq 2 D_{\mu }^{\alpha } (R) \leq 2 D_{\mu }^{\alpha } (R_{0}) 
\]
Thus we deduce
\[
\kappa \cdot \sup_{a \in q(\mathbb{R})}  \lambda ( q^{-1}(a) \cap A_{c,t} ) ^{1+r} \leq 
2 D_{\mu }^{\alpha } (R_{0}).
\]
Also, since $f$ is upper bounded by $t$ on $  A_{c,t}$, 
\begin{eqnarray}
\max_{a \in q(\mathbb{R})}  \mu ( q^{-1}(a)\cap  A_{c,t} ) 
&\leq & t \cdot
\sup_{a \in q(\mathbb{R})}  \lambda ( q^{-1}(a) \cap A_{c,t} )  \nonumber \\
&\leq &
t \left(  \frac{ 2D_{\mu}^{\alpha } (R_{0}) }{\kappa}  \right)^{\frac{1}{1+r}} < \frac{\varepsilon }{2} .
\label{ref_afrmjie_44_77}
\end{eqnarray}
With (\ref{iequ_ac_at_2}) and (\ref{ref_afrmjie_44_77})
we finally obtain
\[
\max_{a \in q(\mathbb{R})}  \mu ( q^{-1}(a) )
\leq \max_{a \in q(\mathbb{R})}  \mu ( q^{-1}(a) \cap A_{c,t} ) + 
1 - \mu ( A_{c,t}  ) < \varepsilon ,
\]
which proves (\ref{ref_iequ_b_01}). 
Now, additionally, let $i(f) > 0$. 
Let $a \in q(\mathbb{R})$ with $\mu ( q^{-1}(a) ) > 0$.
By Lemma~\ref{lemm_mono} we can assume, that $q^{-1}(a)$ is an interval.
Thus we obtain 
\[
\lambda  (q^{-1}(a) \cap I_{i}) = \diam ( q^{-1}(a) \cap I_{i} )
\] 
for every $i \in \{ 1,\ldots,m \}$.
Together with $I_{i} \subset \supp (\mu )$ we deduce
\begin{eqnarray}
\nonumber 
\mu ( q^{-1}(a) ) \geq \mu ( q^{-1}(a) \cap I_{i} ) &=&
\int_{q^{-1}(a) \cap I_{i}} f(x) \, d \lambda  (x) \\ \nonumber 
& \geq & i(f) \lambda  (q^{-1}(a) \cap I_{i}) \\  
&=&  i(f) \diam ( q^{-1}(a) \cap I_{i} ) 
\label{ref_789f_rff_3224}
\end{eqnarray}
for every $i \in \{ 1,\ldots,m \}$.
Relation (\ref{ref_equ_onedim}) follows now immediately from
(\ref{ref_iequ_b_01}) and (\ref{ref_789f_rff_3224}).    \qed

\medskip

\noindent\emph{Proof of Lemma~\ref{pierce_neg_para}.} \ Let $\alpha
\in [ - \infty, 0]$.  Then by Lemma~\ref{lemm_mono} for any $\gamma>1$
there exists a quantizer $q$ with $H_{\mu }^{\alpha }(q) \leq R$ such
that each cell of $q$ is an interval with positive $\mu-$mass (and thus $q(\R)=N_q$) and
\begin{equation}
\label{ref_lemm_dmug}
D_{\mu}(q) \leq \gamma \cdot D_{\mu}^{ \alpha }(R)   .
\end{equation}
According to definition (\ref{ref_def_entr}) we obtain in case of $\alpha > - \infty$ that
\begin{equation}
\label{n_upp_bou}
\sum_{a \in q(\mathbb{R})} 
\mu ( q^{-1}(a) )^{\alpha } \leq e^{(1-\alpha )R}.
\end{equation}
We deduce
\begin{equation}
\label{R_upp_bou}
e^{R} \geq \biggl(\; \sum_{a \in q(\mathbb{R})} \mu ( q^{-1}(a) ) (1/ \mu
( q^{-1}(a) ))^{1-\alpha } \biggr)^{1/(1-\alpha )} \geq n 
\end{equation}
where $n$ is the number of codepoints of $q$ and the second inequality follows
from Jensen's inequality applied to the concave function $x \mapsto
x^{1/(1-\alpha)}$. In case of $\alpha = - \infty$ we  obviously have $n < \infty$. We get
\[
1 = \sum_{a \in q(\mathbb{R})} \mu ( q^{-1}(a) ) 
\geq n \cdot \min \{ \mu ( q^{-1}(a) ) : a \in q(\mathbb{R}) \} \geq n  e^{-R}.
\]
Hence, $n \leq e^{R}$ for every $\alpha \in [-\infty, 0]$. Let
$ \supp (\mu )=[c,d]$, where $-\infty<c<d<\infty$. 
For every $a \in q(\mathbb{R})$ let $m_{a}$ denote  the midpoint of $q^{-1}(a) \cap [c,d]$. 
As in the proof of Lemma~\ref{ref_lemm_gray}, we obtain
\begin{eqnarray*}
D_{\mu }(q) &=& \sum_{a \in q(\mathbb{R})} \int_{q^{-1}(a)} |x-a|^{r}
f(x)\,  d \lambda(x) \\
& \geq & \sum_{a \in q(\mathbb{R})} i(f) \int_{q^{-1}(a) \cap [c,d]}
|x-m_{a}|^{r}\,  d \lambda(x) \\
&=& \sum_{a \in q(\mathbb{R})} i(f) ( 2^{r}(1+r) )^{-1} \diam ( q^{-1}(a) \cap [c,d] )^{r+1} .
\end{eqnarray*}
Clearly (cf.\ (\ref{bou_leng_meas})), 
\[
\frac{\mu (q^{-1}(a))}{s(f)} 
\leq 
\diam (  q^{-1}(a)  \cap [c,d] )
\leq 
\frac{\mu (q^{-1}(a))}{i(f)} 
\] 
for every $a \in q(\mathbb{R})$ with $\mu(q^{-1}(a))>0$.
Thus we deduce
from the convexity of $x \mapsto x^{r+1}$ that
\begin{eqnarray*}
D_{\mu }(q) & \geq & i(f) ( 2^{r}(1+r) )^{-1} s(f)^{-r-1} \sum_{a \in q(\mathbb{R})} ( \mu (q^{-1}(a)) )^{r+1} \\
& \geq & i(f) ( 2^{r}(1+r) )^{-1} s(f)^{-r-1} \sum_{i=1}^{n} (1/n)^{r+1} \\
&=&
i(f) ( 2^{r}(1+r) )^{-1} s(f)^{-r-1} n^{-r}.  
\end{eqnarray*}
Combining (\ref{ref_lemm_dmug}) and Proposition~\ref{ref_lemm_pierce}
(ii)  we obtain 
\begin{equation}
\label{ref_leq_ih}
i(f) ( 2^{r}(1+r) )^{-1} s(f)^{-r-1} n^{-r} \leq \gamma \frac{2^{r}}{i(f)^{r}} e^{-rR}.
\end{equation}
Because $\gamma \in {} (1, \infty)$ was arbitrary, inequality (\ref{ref_leq_ih}) remains valid if we set 
$\gamma = 1$. 
Hence we obtain
\[
\left( \frac{i(f)}{s(f)} \right)^{\frac{r+1}{r}}
\left(\frac{1}{4^{r}(1+r)}\right)^{1/r} e^{R} \leq n , 
\]
which yields (\ref{dmuealphrhr}). \\
Now assume $\alpha \in (- \infty, 0)$. We will modify $q$ such that
the new quantizer is in  $\mathcal{K}_{R}$ and it still satisfies the
rate constraint, while its distortion does not exceed that of $q$.
Let 
\[
p = \max \{ \mu ( q^{-1}(a) ) : a \in q(\mathbb{R})  \} > 0
\]
and $a_{p} \in q(\mathbb{R})$ such that $\mu ( q^{-1}(a_{p}) ) = p$.
If $H_{\mu }^{\alpha }( q ) < R$ we can subdivide the cell $q^{-1}(a_{p})$
into two cells with equal $\mu-$mass, such that the entropy increases by
\begin{eqnarray*}
&& - \frac{1}{1 - \alpha } \log \biggl( p^{\alpha } 
+ \sum_{a \in q(\mathbb{R}) \setminus \{ a_{p} \} }  \mu ( q^{-1}(a) )^{\alpha }  \biggr) \\
&& \quad +
\frac{1}{1 - \alpha } \log \biggl( 2 (p/2)^{\alpha } 
+ \sum_{a \in q(\mathbb{R}) \setminus \{ a_{p} \} }  \mu ( q^{-1}(a) )^{\alpha }  \biggr) > 0.
\end{eqnarray*}
If we take the optimal quantization points for the two new cells, the
new quantizer does not increase the quantization error. As long as the entropy
is lower than $R$ we repeat this procedure. Hence there exists a modified
quantizer (also denoted by $q$) satisfying
\begin{equation}
\label{upp_bou_num_cell}
e^{(1-\alpha )R} - e^{(1-\alpha ) H_{\mu }^{\alpha }(  q  )} \leq ( 2^{1-\alpha } - 1 ) p^{\alpha }.
\end{equation}
Note that (\ref{n_upp_bou}) and (\ref{R_upp_bou}) remain valid also for this
modified quantizer.
Consequently, 
\begin{equation}
\label{bound_tild}
0 < C e^{R} \leq \card (   q ) \leq  e^{R} < \infty .
\end{equation}
Thus we deduce 
\[
e^{(1- \alpha )R} \geq e^{(1- \alpha )H_{\mu}^{\alpha }(  q )} =
\sum_{a \in q(\mathbb{R})} \mu (  q^{-1}(a) )^{\alpha } \geq
\card (   q )  \cdot p^{\alpha} \geq Ce^{R} p^{\alpha },
\]
which implies
\[
p^{\alpha } \leq C^{-1} e^{-R} e^{(1- \alpha )R}.
\]
Together with (\ref{upp_bou_num_cell}) and $R> \log ( \frac{2^{1-\alpha }-1}{C} )$ we obtain 
\begin{eqnarray}
1 &\leq & \frac{e^{(1- \alpha )R}}{ e^{(1-\alpha )H_{\mu }^{\alpha }(  q )} } 
= \frac{e^{(1- \alpha )R}}{e^{(1- \alpha )R} 
- ( e^{(1- \alpha )R} -  e^{(1-\alpha )H_{\mu }^{\alpha }(  q )} )} \nonumber \\
& \leq & 
\frac{e^{(1- \alpha )R}}{e^{(1- \alpha )R} -  ( 2^{1-\alpha } - 1 ) p^{\alpha } } \nonumber \\
& \leq &
\frac{e^{(1- \alpha )R}}{e^{(1- \alpha )R} 
-  ( 2^{1-\alpha } - 1 ) C^{-1} e^{-R} e^{(1- \alpha )R} } .
\label{bisec_bou}
\end{eqnarray}
In view of (\ref{bisec_bou}) and (\ref{bound_tild}) we conclude that $
q \in \mathcal{K}_{R}$, 
which proves (\ref{dmuealphrkr}).  \qed

\begin{proof}[Proof of Lemma \ref{lemm_inf_g}]
Recall the definition  (\ref{const_quant}) of constant   $C$.
Fix $\kappa \in (0,C)$. Let $R_{0} > 0$ such that $Ce^{R}-(m-1) \geq \kappa
e^{R}$ for every $R \geq R_{0}$. 
According Lemma~\ref{ref_lemm_gray}, in the definition of
$D_{\mu}^{ \alpha }(R)$ 
it suffices w.l.o.g.\  to
consider for $R \geq R_{0}$ only those quantizers $q \in \mathcal{H}_{R}$ 
satisfying
\begin{equation}
\label{ref_equ_diamf1q}
\sup \{ \diam ( q^{-1}(a) \cap I ) : a \in q(\mathbb{R}) \} < \diam (I) / 2m .
\end{equation}
In view of Lemma~\ref{pierce_neg_para} 
it suffices to show that for $R\ge R_{0}$ any quantizer  $q \in
\mathcal{H}_{R}$  that satisfies (\ref{ref_equ_diamf1q}) can be
modified such that the distortion of the new quantizer $\tilde{q}$ does
not exceed that of $q$ and it  satisfies  $\tilde{q} \in
\mathcal{K}_{R}(\kappa)$ and 
\[
2 \inf \{ \diam ( \tilde{q}^{-1}(a) \cap I ) : a \in S(\tilde{q}) \} \geq \inf \{ \diam ( \tilde{q}^{-1}(a) ) : a \in A(\tilde{q}) \} .
\]

According to the
upper bound (\ref{ref_equ_diamf1q}) we always have $A(q) \neq
\emptyset$.  If $S(q)= \emptyset$, then the assertion is
obvious. Hence, let $S(q) \neq \emptyset$.  Let us assume
w.l.o.g.\ that $\mu ( q^{-1}(b) ) > 0$ and that (see Lemma~\ref{lemm_mono}) 
$q^{-1}(b)$ is an interval for every $b \in
q(\mathbb{R})$.  For every $a \in S(q)$ let $\emptyset \neq N(a)
\subset q(\mathbb{R}) \setminus \{ a \}$ be the set of neighbor
points, i.e., for every $b \in N(a)$ we have either $\sup
q^{-1}(b)=\inf q^{-1}(a)$ or $\inf q^{-1}(b)=\sup q^{-1}(a)$.  Due to
(\ref{ref_equ_diamf1q}) we know that $N(a) \cap S(q) = \emptyset$.
Moreover, $N(a) \subset A(q)$ and $\card (N(a))=2$.  Fix $i_{a} \in \{
1,\ldots,m-1 \}$ such that $q^{-1}(a) \subset I_{i_{a}} \cup
I_{i_{a}+1}$.  Because $a \in S(q)$, we have $\Delta_{1} = \diam (
q^{-1}(a) \cap I_{i_{a}} ) > 0$ and $\Delta_{2} = \diam ( q^{-1}(a)
\cap I_{i_{a}+1} ) > 0$.  Moreover, $\diam ( q^{-1}(a)) = \Delta_{1} +
\Delta_{2}$.  Let $b_{1} \in N(a)$ such that $\inf (q^{-1}(a))=\sup
(q^{-1}(b_{1}))$ and let $b_{2} \in N(a)$ such that $\inf
(q^{-1}(b_{2}))=\sup (q^{-1}(a))$.  Next we will show that
\begin{equation}
\label{iequ_delta_neg_01}
\inf (q^{-1}(a)) + \frac{\Delta_{1}}{2} \leq a \leq \inf (q^{-1}(a)) + \Delta_{1} + \frac{\Delta_{2}}{2}.
\end{equation}
To see this, one recognizes that $a$ has to be optimal for $\mu(\cdot | q^{-1}(a))$. 
As a consequence (see, e.g.,  \cite[Lemma 2.6 (a)]{GrLu00}), $a \in [ \inf(q^{-1}(a)), \sup(q^{-1}(a)) ]$.
Moreover, $a$ has to be a stationary point (see \cite[Lemma 2.5]{GrLu00}), which yields 
\begin{equation}
\label{iequ_delta_neg_02}
\int_{[ \inf (q^{-1}(a)), a]} | x-a |^{r-1}\, d \mu (x) = \int_{[ a, \sup (q^{-1}(a)) ]} | x-a |^{r-1}\, d \mu (x).
\end{equation}
Now let us assume that the first inequality in (\ref{iequ_delta_neg_01})
does not hold.
Hence, 
\begin{equation}
\label{lhs_ind}
a < \inf ( q^{-1}(a) ) + \Delta_{1}/2.
\end{equation}
Note that $\sup ( I_{i_{a}} ) = \inf q^{-1}(a) + \Delta_{1}$
and that $\sup ( I_{i_{a}} ) + \Delta_{2} = \sup (q^{-1}(a))$.
From (\ref{iequ_delta_neg_02}) and $\Delta_{2}>0$ we get
\begin{eqnarray}
\label{lhs_ind_trunc}
\int_{[ \inf (q^{-1}(a)), a]} | x-a |^{r-1}\, d \mu (x) 
&>& \int_{[ a, \inf q^{-1}(a) + \Delta_{1} ]} | x-a |^{r-1}\,  d \mu (x).
\end{eqnarray}
Because the density of $\mu$ is constant on $[ \inf (q^{-1}(a)) , \inf (q^{-1}(a)) + \Delta_{1} ]$ we
obtain from (\ref{lhs_ind_trunc}) that $a > \inf ( q^{-1}(a) ) + \Delta_{1}/2$, which contradicts 
(\ref{lhs_ind}). Thus we have proved the 
left inequality in (\ref{iequ_delta_neg_01}).
Similarly, we deduce from
$\Delta_{1}>0$ and (\ref{iequ_delta_neg_02}) the right inequality in
(\ref{iequ_delta_neg_01}).

Recall that $\mu$ has constant density on $q^{-1}(b_{i});$ $i=1,2.$
Again by stationarity (\ref{iequ_delta_neg_02}) we obtain 
\begin{equation}
\label{ref_equ_b1}
b_{1} = \inf q^{-1}(a) - \diam ( q^{-1}(b_{1}) )/2
\end{equation}
and
\[
b_{2} = \inf q^{-1}(a) + \Delta_{1} + \Delta_{2} + \diam ( q^{-1}(b_{2}) )/2 .
\]
Let  $\Delta = \Delta_{1} + \Delta_{2}$. 
Next we show  that  w.l.o.g.\  we can assume $ 2 \Delta \geq \min ( \diam (q^{-1}(b_{1})), \diam
(q^{-1}(b_{2})) )$. Assume to the contrary that 
\begin{equation}
\label{ref_delta_fgh}
2 \Delta < \min ( \diam (q^{-1}(b_{1})), \diam (q^{-1}(b_{2})) ).
\end{equation}
Then we have  $\diam (q^{-1}(b_{1})) > 2\Delta> 2 \Delta_{1} + \Delta_{2}$, and
applying (\ref{ref_equ_b1}) we get
\[
 \inf q^{-1}(a) -b_{1} >   \inf q^{-1}(a) + \Delta_{1} +
 \frac{\Delta_{2}}{2}  -  \inf q^{-1}(a) 
\]
Hence, (\ref{iequ_delta_neg_01}) implies
\begin{equation}
\label{iequ_delta_neg_03}
\inf q^{-1}(a) - b_{1}  >  a - \inf q^{-1}(a).
\end{equation}
Similarly we obtain
\begin{equation}
\label{iequ_delta_neg_04}
b_{2} - \sup q^{-1}(a) >  \sup q^{-1}(a) - a.
\end{equation}
In view of (\ref{ref_delta_fgh}) and by the definition of $\mu$ we have 
\[
\frac{\diam (I)}{m } \, \mu ( q^{-1}(b_{1}) ) = \diam ( q^{-1}(b_{1}) ) \cdot s_{i_{a}} > 2 \Delta \cdot s_{i_{a}}
\]
and
\[ 
\quad \frac{\diam (I)}{m }\, \mu ( q^{-1}(b_{2}) ) = 
\diam ( q^{-1}(b_{2}) ) \cdot s_{i_{a}+1} > 2 \Delta \cdot s_{i_{a}+1}. 
\]
Moreover,
\[
 \frac{\diam (I)}{m }\, \mu ( q^{-1}(a) ) = \Delta_{1} s_{i_{a}} + \Delta_{2} s_{i_{a}+1} < 2 \Delta \max \{ s_{i_{a}} , s_{i_{a}+1} \} .
\]
Thus we obtain 
\begin{equation}
\label{eq:max}
 \mu (q^{-1}(a)) < \max ( \mu (q^{-1}(b_{1})), \mu
(q^{-1}(b_{2})) )
\end{equation}
as long as (\ref{ref_delta_fgh}) holds.  Thus, in view of
(\ref{iequ_delta_neg_03}) and (\ref{iequ_delta_neg_04}), we can modify
$q$ by increasing the codecell $q^{-1}(a)$, which yields a reduction
of the quantization error and a non-increasing entropy of $q$ (due to
$\alpha < 0$, as long as (\ref{eq:max}) holds, the entropy is a
non-decreasing function of the left endpoint of the cell $q^{-1}(a)$
and a non-increasing function of the right endpoint of $q^{-1}(a)$).
The codecell can be expanded this way until $2 \Delta= \min (
\diam (q^{-1}(b_{1})), \diam (q^{-1}(b_{2})) )$ holds. 
Note that independent of this modification $q$ remains an element of
$\mathcal{H}_{R}$.   Thus we can assume w.l.o.g.\ that
\begin{equation}
\label{Delta_inf}
2 \Delta \geq \min ( \diam (q^{-1}(b_{1})), \diam (q^{-1}(b_{2})) ).
\end{equation}
If $q \in \mathcal{K}_{R}(\kappa)$, then the proof is finished.
Hence, let us assume that $q \notin \mathcal{K}_{R}(\kappa)$.
We will show that $q$ can always be modified such that the new  quantizer
belongs to $\mathcal{K}_{R}(\kappa)$ and still satisfies relation
(\ref{Delta_inf}). 
We proceed as in the proof of relation (\ref{dmuealphrkr}).
Let 
\[
W(q) = \{ a \in q( \mathbb{R} ) : \mu (q^{-1}(a)) = \max \{ \mu(q^{-1}(b)) : b \in A(q) \} \}.
\]
We subdivide one by one the cells $q^{-1}(a)$ with $a \in W(q)$
and $p=\mu (q^{-1}(a))$ as in the proof of (\ref{dmuealphrkr}) in
Lemma~\ref{pierce_neg_para}.  Note, that the entropy of the
quantizer will exceed any given bound if we repeat the subdivision
process enough times.  We stop this process with a quantizer $\tilde
q$ that satisfies relation (\ref{upp_bou_num_cell}).  Now recall that
$ Ce^{R}-(m-1)\ge \kappa e^R$ if $R\ge R_0$ by the definition at the
beginning of the proof. Thus, with $p=\mu (\tilde{q}^{-1}(a))$, we
have 
\begin{eqnarray*}
e^{(1-\alpha )R} 
& \geq & e^{(1-\alpha )H_{\mu }^{\alpha }(\tilde q)} \geq ( \card (\tilde q(\mathbb{R})) - (m-1) ) p^{\alpha } \\
& \geq &
(Ce^{R}-(m-1))p^{\alpha } \geq  \kappa e^{R}p^{\alpha }.
\end{eqnarray*}
Now the  inequality  $e^{(1-\alpha )R} \ge  me^{R}p^{\alpha }$ 
allows us to perform steps identical to the ones in the chain of
inequalities   (\ref{bisec_bou}) 
and we obtain  that the  quantizer belongs to $\mathcal{K}_{R}(\kappa)$.
Obviously, (\ref{Delta_inf}) is still in force for $\tilde q$ and the
proof is complete.
\end{proof}

\appsec
\label{appC}

\noindent\emph{Proof of Proposition~\ref{ref_pro_zerleg}.} \ For every
$i \in \{1,\ldots,m\}$ choose a quantizer $q_{i} \in \mathcal{Q}$ for
$\mu_{i}$ with $H_{\mu_{i}}^{\alpha }(q_{i}) \leq R_{i}$.  Let
\[
J_{i} = \{ a \in q_{i}(\mathbb{R}) : \mu ( q_{i}^{-1}(a) \cap A_{i} ) > 0 \}. 
\]
Let $I_{i} \subset \mathbb{N}$ be an index set of the same 
cardinality as $J_{i}$ and for every $k \in I_{i}$ choose $a_{i,k} \in
q_{i}(\mathbb{R})$  
such that $J_{i}=\{ a_{i, k} : k \in I_{i} \}$.
Let 
\[
N = \mathbb{R} \setminus \cup_{i = 1}^{m} \cup_{k \in I_{i}} q_{i}^{-1}(a_{i, k}) \cap A_{i}.
\]
Note that $\mu (N)=0$.
Now we define the quantizer $q$ by 
the codecells 
\[
\{ N \} \cup \{ q_{i}^{-1}(a_{i, k}) \cap A_{i} : i=1,\ldots,m ; k \in I_{i}  \}
\]
and corresponding codepoints
\[
\{ 0 \} \cup \{ a_{i, k} : i=1,\ldots,m ; k \in I_{i} \} .
\]
Note that despite our general assumption,  the codepoints now  are not necessarily distinct.
Recall the convention $0^{0}=0$. Since  $\mu(N)=0$,  the definition of $H_{\mu }^{\alpha } (q)$ yields
\begin{eqnarray*}
H_{\mu }^{\alpha } (q) &=& 
\frac{1}{1-\alpha } \log \left(  \sum_{i=1}^{m} \sum_{k \in I_{i}} 
\mu ( q_{i}^{-1}(a_{i, k}) \cap A_{i} ) ^{\alpha }
\right) \\
&=& 
\frac{1}{1-\alpha } \log \left(  \sum_{i=1}^{m} s_{i}^{\alpha }\sum_{a \in q_{i}(\mathbb{R})} 
\mu_{i} ( q_{i}^{-1}(a) ) ^{\alpha }
\right) \\
&=& \frac{1}{1-\alpha } \log \left( 
\sum_{i=1}^{m} s_{i}^{\alpha } e^{(1-\alpha ) H_{\mu_{i}}^{\alpha }(q_{i}) } \right). 
\end{eqnarray*}
Since  $H_{\mu_{i}}^{\alpha }(q_{i}) \leq R_{i}$,  we obtain in both cases ($\alpha < 1$ and $\alpha > 1$)
that
\[
H_{\mu }^{\alpha } (q) \leq \frac{1}{1-\alpha } \log \left( 
\sum_{i=1}^{m} s_{i}^{\alpha } e^{(1-\alpha ) R_{i} } \right).
\]
Now it is easy to check that $H_{\mu }^{\alpha } (q) \leq R$ is satisfied if either
(\ref{ref_cond_01}) or (\ref{ref_cond_02}) holds. 
Further we deduce
\begin{eqnarray*}
D_{\mu }^{\alpha } (R) &\leq & D_{\mu}(q) = \int | x - q(x) | ^{r} \, d \mu (x) \\
&=& \sum_{i=1}^{m} s_{i} \int_{A_{i}} | x - q_{i}(x) | ^{r} \ d \mu_{i} (x) 
= \sum_{i=1}^{m} s_{i} D_{\mu_{i}} (q_{i}).  
\end{eqnarray*}
Taking the infimum on the right hand side of above inequality yields
the assertion. \qed 

\medskip

\noindent\emph{Proof of Lemma~\ref{ref_lemm_smu12}.} \ 
From Definition \ref{ref_def_mdivis} we have $s \in (0,1)$.
Let $R \geq 0$ and $\delta > 0$.  
Let $q \in \mathcal{Q}$ with $H_{\mu }^{\alpha }(q) \leq R$ and $\delta + D_{\mu }^{\alpha }(R) \geq D_{\mu }(q)$.
We obtain
\begin{equation}
\label{ref_equ_dmualprdmg}
\delta + D_{\mu }^{\alpha }(R) \geq D_{\mu }(q) \geq s \int | x - q(x)
|^{r}\,  d \mu_{i_{0}} (x).
\end{equation}
Since $\alpha \in [0,1)$, we deduce
\begin{eqnarray*}
R  & \geq & H_{\mu }^{\alpha } (q) = \frac{1}{1-\alpha } \log 
\left( \sum_{a \in q( \mathbb{R} )} 
\left( \sum_{i=1}^{m} s_{i} \mu_{i} ( q^{-1}(a) )  \right)^{\alpha } \right) \\
& \geq & \frac{1}{1-\alpha } \log \biggl( \; \sum_{a \in q( \mathbb{R} )}
\left( s \mu_{i_{0}} ( q^{-1}(a) )  \right)^{\alpha } \biggr) \\
&=& \frac{\alpha }{1-\alpha } \log (s) + H_{\mu_{i_{0}}}^{\alpha } (q).
\end{eqnarray*}
Because $\delta$ was arbitrary we get from (\ref{ref_equ_dmualprdmg}) that
\[
D_{\mu }^{\alpha } (R) \geq s D_{\mu_{i_{0}}}^{\alpha } 
\left( R - \frac{\alpha }{1-\alpha } \log (s) \right),
\]
which yields
\begin{eqnarray*}
e^{r R} D_{\mu }^{\alpha }(R) & \geq &
s e^{ r \left( \frac{\alpha }{1-\alpha } \log (s) \right) }
e^{ r \left( R - \frac{\alpha }{1-\alpha } \log (s) \right) }  D_{\mu_{i_{0}}}^{\alpha } 
\left( R - \frac{\alpha }{1-\alpha } \log (s) \right)  \\
&=&  s^{a_{1} a_{2}} e^{r \left( R - \frac{\alpha }{ 1- \alpha }   
\log (s) \right)} D_{\mu_{i_{0}}}^{\alpha }\left( R - \frac{\alpha }{ 1-
  \alpha } \log (s) \right) 
\end{eqnarray*}
and therefore proves (\ref{ref_equ_ergnjks}).

Now let $\alpha \in [0,r+1 ) \setminus \{ 1 \}$ and 
fix $R_{0}> 0$, such that
\[
R_{0} \geq \max \{ -\log( t_{i} ): i=1,\ldots,m \}.
\] 
For any $R > R_{0}$ let $R_{i} = R + \log( t_{i} ) > 0$, $i=1,\ldots,m$.
We obtain
\begin{equation}
\label{ref_iequ_sie}
\sum_{i=1}^{m} s_{i}^{\alpha} e^{(1-\alpha )R_{i}} = 
e^{( 1-\alpha )R},
\end{equation}
if $\alpha \in [0,r+1 ) \setminus \{ 1 \}$.
Indeed, (\ref{ref_iequ_sie}) is equivalent to
$ \sum_{i=1}^{m} s_{i}^{\alpha} t_{i}^{1-\alpha } = 1$.
But this equation is satisfied by the definition of $t_{i}$. 
Applying Proposition~\ref{ref_pro_zerleg} we obtain
\[
D_{\mu }^{\alpha } (R) \leq s D_{\mu_{i_{0}} }^{\alpha } \left( R_{i_{0}} \right)
+ \sum_{i=1; i \neq i_{0}}^{m} s_{i} D_{\mu_{i} }^{\alpha } \left( R_{i} \right).
\]
Thus we can compute
\begin{eqnarray}
e^{rR} D_{\mu }^{\alpha }(R) & \leq & 
e^{rR} s D_{\mu_{i_{0}} }^{\alpha } ( R_{i_{0}} ) + \sum_{i=1; i \neq i_{0}}^{m}
e^{rR} s_{i} D_{\mu_{i} }^{\alpha } ( R_{i} ) \nonumber \\
&=& 
e^{r(R-R_{i_{0}})} s e^{rR_{i_{0}}} D_{\mu_{i_{0}} }^{\alpha } ( R_{i_{0}} ) + \sum_{i=1; i \neq i_{0}}^{m}
e^{r(R-R_{i})} s_{i} e^{rR_{i}} D_{\mu_{i} }^{\alpha } ( R_{i} ) \nonumber \\
&=& s t_{i_{0}}^{-r} e^{rR_{i_{0}}} D_{\mu_{i_{0}} }^{\alpha } ( R_{i_{0}} ) + \sum_{i=1; i \neq i_{0}}^{m}
s_{i} t_{i}^{-r} e^{rR_{i}} D_{\mu_{i} }^{\alpha } ( R_{i} ).
\label{r_h_s_363}
\end{eqnarray}
Because all terms in  (\ref{r_h_s_363}) are nonnegative  
we obtain (\ref{ref_iequ_rrinferr}). \qed

\appsec
\label{appD}

\begin{lemma}
\label{lemma_straddle_conv}
Let $m \in \mathbb{N}$ and $\sum_{i=1}^{m}s_{i} = 1$ with $s_{i}>0$ for every $i \in \{1,\ldots,m\}$.
Let the probability measure $\mu$ be supported on a bounded interval $I$
such that
$\mu = \sum_{i=1}^{m} s_{i} U(I_{i})$  where the $I_{i}$ are intervals
of  equal length $\lambda(I)/m$ that partition  $I$. 
Let $\alpha \in (- \infty, 0)$ and $(R_{n})_{n \in \mathbb{N}}$ be an
increasing sequence of positive numbers such that 
$R_{n} \to \infty$ as $n \to \infty$. Then for every sequence 
$(q_{n})_{n \in \mathbb{N}}$ of quantizers with $q_{n} \in
\mathcal{G}_{R_{n}}$, relation (\ref{ref_equ_deltedas}) holds. 
\end{lemma}
\begin{proof}
Recall from (\ref{aqsqdef}) the definition of $A(q)$ and $S(q)$. For
any  $n \in \mathbb{N}$ 
\begin{eqnarray}
1 & \leq & \frac{e^{(1- \alpha )H_{\mu}^{\alpha }(q_{n})}}{\sum_{a \in A(q_{n})} 
\mu ( q_{n}^{-1}(a) )^{\alpha } } \nonumber \\
&=& \frac{\sum_{a \in A(q_{n})} \mu ( q_{n}^{-1}(a) )^{\alpha } + 
\sum_{a \in S(q_{n})} \mu ( q_{n}^{-1}(a) )^{\alpha } }
{\sum_{a \in A(q_{n})} \mu ( q_{n}^{-1}(a) )^{\alpha } } \nonumber \\
& \leq &
1 + \frac{\card ( S (q_{n}) ) \cdot \sup \{ \mu ( q_{n}^{-1}(a) )^{\alpha } : a \in S(q_{n}) \}}
{\sum_{a \in A(q_{n})} \mu ( q_{n}^{-1}(a) )^{\alpha } } \nonumber \\
&\leq & 1 + (m-1) \frac{ \sup \{ \mu ( q_{n}^{-1}(a) )^{\alpha } : a \in S(q_{n}) \}}
{\sum_{a \in A(q_{n})} \mu ( q_{n}^{-1}(a) )^{\alpha } } \nonumber \\
&=& 1 + (m-1) \frac{ (\inf \{ \mu ( q_{n}^{-1}(a) ) : a \in S(q_{n}) \}) ^{\alpha }}
{\sum_{a \in A(q_{n})} \mu ( q_{n}^{-1}(a) )^{\alpha } }.
\label{upp_bou_01}
\end{eqnarray}
Now let 
\[
h_{1}=\min \biggl\{ \frac{s_{i}}{\lambda(I)/m} : i \in \{ 1,\ldots,m \} \biggr\} > 0
\]
and
\[
h_{2}=\max \biggl\{ \frac{s_{i}}{\lambda(I)/m} : i \in \{ 1,\ldots,m \} \biggr\} > 0 .
\]
Since 
$q_{n} \in \mathcal{G}_{R_{n}}$,  we have 
\begin{eqnarray}
1 & \leq & 1 + (m-1)(h_{1}/2)^{\alpha } \frac{ (\min \{ \diam ( q_{n}^{-1}(a) ) : a \in A(q_{n}) \}) ^{\alpha }}
{\sum_{a \in A(q_{n})} \mu ( q_{n}^{-1}(a) )^{\alpha } } \nonumber \\
& \leq & 1 + (m-1)(h_{1}/2h_{2})^{\alpha } \frac{ (\min \{ \diam ( q_{n}^{-1}(a) ) : a \in A(q_{n}) \}) ^{\alpha }}
{\sum_{a \in A(q_{n})} \diam ( q_{n}^{-1}(a) )^{\alpha } } .
\label{upp_bou_02}
\end{eqnarray}
Fix $i=i(n)\in \{1,\ldots,m \}$ and $b \in A(q_{n}) \cap I_{i}$ such that 
\begin{equation}
\label{min_def_point}
\diam ( q_{n}^{-1}(b) ) =  \min \{ \diam ( q_{n}^{-1}(a) ) : a \in A(q_{n}) \}.
\end{equation}
From Proposition~\ref{ref_prop_unit_cube} and by \cite[Example 5.5]{GrLu00} 
we know that all codecells $q_{n}^{-1}(a)$ with $a \in A(q_{n}) \cap
I_{i}$ can be assumed to have equal length.
Because $q_{n} \in \mathcal{G}_{R_n} \subset \mathcal{K}_{R_n}$ we 
obtain $\lim_{n \rightarrow \infty} H_{\mu }^{\alpha }(q_{n}) = \infty$.
In view of (\ref{min_def_point}) we thus get $\card ( A(q_{n}) \cap I_{i} ) \to \infty$ as $n \to \infty$.
From (\ref{upp_bou_01}) and (\ref{upp_bou_02}) we deduce
\begin{eqnarray}
1 & \leq & \frac{e^{(1- \alpha )H_{\mu}^{\alpha }(q_{n})}}{\sum_{a \in A(q_{n})} \mu ( q_{n}^{-1}(a) )^{\alpha } } 
\leq  1 +  \frac{ (m-1)(h_{1}/2h_{2})^{\alpha }(\diam (q_{n}^{-1}(b))) ^{\alpha }} 
{\sum_{a \in A(q_{n}) \cap I_{i}} \diam ( q_{n}^{-1}(a) )^{\alpha } } \nonumber \\
&=& 1 +  \frac{ (m-1)(h_{1}/2h_{2})^{\alpha }} 
{ \card ( A(q_{n}) \cap I_{i} ) }
\to 1 \text{ as } n \to \infty .
\label{rel_conv_entr}
\end{eqnarray}
Again from $q_{n} \in \mathcal{G}_{R_{n}} \subset \mathcal{K}_{R_{n}}$ we have 
$\lim_{n \rightarrow \infty} e^{(1 - \alpha )(R_{n} - H_{\mu}^{\alpha }(q_{n}))} = 1$, 
which yields together with (\ref{rel_conv_entr})
the assertion. 
\end{proof}

Let $m \geq 2$ and $s_{1},\ldots,s_{m} \in (0,1)^{m}$ with
$\sum_{i=1}^{m}s_{i}=1$. For $(v_{1},\ldots,v_{m}) \in (0,\infty )^{m}$
and $\alpha \in ( - \infty, \infty ) \setminus \{ 1 \}$
we define
\[
F(v_{1},\ldots,v_{m}) = \sum_{i=1}^{m} s_{i} v_{i}^{-r}
\]
and set $t_{i} = s_{i}^{1/a_2} \left( \sum_{j=1}^{m} s_{j}^{
  a_1} \right)^{-\frac{1}{1-\alpha }}$, $i=1,\ldots,m$ as in
(\ref{ref_def_ti}).

\begin{lemma}
\label{ref_lem_lower_bound}
If $\alpha \in ( -\infty, 1)$, then
\[
F(t_{1},\ldots,t_{m}) = \inf \{ F(v_{1},\ldots,v_{m}) : (v_{1},\ldots,v_{m}) \in {} (0,\infty )^{m}; \,
\sum_{i=1}^{m} s_{i}^{\alpha } v_{i}^{1-\alpha } = 1 \}.
\]
\end{lemma}

\begin{proof}
Let $x_{i} = s_{i}^{\alpha } v_{i}^{1- \alpha }$. We calculate
\[
v_{i} = ( x_{i} s_{i}^{- \alpha } )^{\frac{1}{1 - \alpha }}
\]
and
\begin{eqnarray*}
F(v_{1},\ldots,v_{m}) & = & \sum_{i=1}^{m} s_{i} ( x_{i} s_{i}^{- \alpha } )^{\frac{-r}{1 - \alpha }} \\
&=& \sum_{i=1}^{m} s_{i}^{\frac{1 - \alpha + \alpha r }{1 - \alpha }} x_{i}^{- \frac{r}{1 - \alpha }} 
=: G(x_{1},\ldots,x_{m}).
\end{eqnarray*}
Applying \cite[Lemma 6.8]{GrLu00} we deduce that $G$ attains its 
minimum on $(0, \infty )^{m}$ subject to the constraint $\sum_{i=1}^{m}x_{i} = 1$
at the point $(y_{1},\ldots,y_{m})$ with
\[
y_{i} = \frac{ \left( s_{i}^{\frac{1 - \alpha + \alpha r }{1 - \alpha }} \right)^{\frac{1}{1+\frac{r}{1-\alpha }}} }
{ \sum_{j=1}^{m} \left( s_{j}^{\frac{1 - \alpha + \alpha r }{1 - \alpha }} \right)^{\frac{1}{1+\frac{r}{1-\alpha }}}  }
= \frac{s_{i}^{a_1}}{\sum_{j=1}^{m}s_{j}^{a_1}}
\]
for every $i \in \{ 1,\ldots,m \}$. Hence, $F$ attains its minimum subject to the constraint
$\sum_{i=1}^{m} s_{i}^{\alpha } v_{i}^{1-\alpha } = 1$ at the point
$(w_{1},\ldots,w_{m})$ with $w_{i} = (y_{i} s_{i}^{- \alpha })^{\frac{1}{1 - \alpha }}$ for 
every $i \in \{ 1,\ldots,m \}$.
We deduce
\begin{eqnarray*}
w_{i}^{1 - \alpha } &=& 
\frac{ s_{i}^{\frac{1 - \alpha + \alpha r}{1 - \alpha + r}} s_{i}^{- \alpha } }
{\sum_{j=1}^{m} \left( s_{j}^{\frac{1 - \alpha + \alpha r }{1 - \alpha }} \right)^{\frac{1}{1+\frac{r}{1-\alpha }}}}
= 
\frac{s_{i}^{ \frac{(1-\alpha )^{2}}{1 - \alpha + r}}}
{\sum_{j=1}^{m}s_{j}^{a_1}}  
\end{eqnarray*}
which yields $w_{i}=t_{i}$.
\end{proof}

\section*{Acknowledgments}
The authors would like to thank two anonymous reviewers for their
detailed and constructive comments.


\end{document}